\documentclass[nohyperref]{article}
\PassOptionsToPackage{numbers, compress}{natbib}
\usepackage[accepted]{icml2025}

\usepackage[utf8]{inputenc} %

\usepackage{graphicx}

\usepackage[T1]{fontenc}

\usepackage[hidelinks,colorlinks]{hyperref}       %
\usepackage{url}            %
\usepackage{booktabs}       %
\usepackage{amsfonts}       %
\usepackage{nicefrac}       %
\usepackage{microtype}      %
\usepackage{amssymb}
\usepackage{pifont}

\IfFileExists{headers/config/showoverfull.config}{
	\overfullrule=1cm
}{
}

\usepackage{marginnote}

\usepackage[backgroundcolor=none,linecolor=red,textsize=footnotesize]{todonotes}

\usepackage{etoolbox}

\newbool{includeappendix}
\setbool{includeappendix}{true} %
\IfFileExists{headers/config/noappendix.config}{
	\setbool{includeappendix}{false}
}{}

\newif\ifincludeappendixx
\ifbool{includeappendix}{
	\includeappendixxtrue
}{
	\includeappendixxfalse
}

\usepackage{xr} %
\usepackage{filecontents}

\ifbool{includeappendix}{}{
	\input{appendix-labels-loader}

	\externaldocument{appendix-labels}
}

\newcommand{\eg}{e.g., }
\newcommand{\ie}{i.e., }

\newcommand{\wrt}{{w.r.t.\ }}

\newcommand{\fpf}{FP4}

\usepackage{acro} %

\DeclareAcronym{cli} {
    short = CLI,
    long = Command Line Interface,
}

\usepackage{listings}

\usepackage{textcomp}

\usepackage{xcolor}

\usepackage[scaled=0.8]{beramono}

\definecolor{ckeyword}{HTML}{7F0055}
\definecolor{ccomment}{HTML}{3F7F5F}
\definecolor{cstring}{HTML}{2A0099}

\lstdefinestyle{numbers}{
	numbers=left,
	framexleftmargin=20pt,
	numberstyle=\tiny,
	firstnumber=auto,
	numbersep=1em,
	xleftmargin=2em
}

\lstdefinestyle{layout}{
	frame=none,
	captionpos=b,
}

\lstdefinestyle{comment-style}{
	morecomment=[l]//,
	morecomment=[s]{/*}{*/},
	commentstyle={\color{ccomment}\itshape},
}

\lstdefinestyle{string-style}{
	morestring=[b]",%
	morestring=[b]',%
	stringstyle={\color{cstring}},
	showstringspaces=false,%
}

\lstdefinestyle{keyword-style}{
	keywordstyle={\ttfamily\bfseries},
	morekeywords={
		function,
		constructor,
		int,
		bool,
		return,
		returns,
		uint
	},
	morekeywords = [2]{},
	keywordstyle = [2]{\text},
	sensitive=true,
}

\lstdefinestyle{input-encoding}{
	inputencoding=utf8,
	extendedchars=true,
	literate=
	{ℝ}{$\reals$}1%
	{→}{$\rightarrow$}1%
	{α}{$\alpha$}1%
	{β}{$\beta$}1%
	{λ}{$\lambda$}1%
	{θ}{$\theta$}1%
	{ϕ}{$\phi$}1%
}

\lstdefinestyle{escaping}{
	moredelim={**[is][\color{blue}]{\%}{\%}},
	escapechar=|,
	mathescape=true
}

\lstdefinestyle{default-style}{
	basicstyle=\fontencoding{T1}\ttfamily\footnotesize,
	style=numbers,
	style=layout,
	style=comment-style,
	style=string-style,
	style=keyword-style,
	style=input-encoding,
	style=escaping,
	tabsize=2,
	upquote=true
}

\lstdefinelanguage{BASIC}{
	language=C++,
	style=default-style
}[keywords,comments,strings]%

\newcommand{\circled}[1]{\raisebox{.5pt}{\textcircled{\raisebox{-.9pt} {#1}}}}

\newcommand{\QK}[2][]{%
  \ensuremath{Q#2\_{K%
    \if\relax\detokenize{#1}\relax
    \else \_{#1}%
    \fi}}%
}

\usepackage[utf8]{inputenc} %
\usepackage[T1]{fontenc}    %
\usepackage{algcompatible}
\usepackage{caption}
\usepackage{algpseudocode}
\usepackage{enumerate}
\usepackage{xurl}            %
\usepackage{booktabs}       %
\usepackage{amsfonts}       %
\usepackage{nicefrac}       %
\usepackage{microtype}      %
\usepackage{xcolor}         %
\usepackage[ruled,vlined]{algorithm2e}
\usepackage{tikz,booktabs,multirow,enumitem,bm}
\usepackage{colortbl}
\usepackage{tabularx}
\usepackage{arydshln}
\usepackage{subcaption}
\usepackage{wrapfig}
\usepackage{tabulary}
\usepackage{amsmath,amsthm,amsfonts}

\usepackage{amsmath,amsfonts,bm}
\usepackage{amsthm}
\usepackage{mathtools}

\newtheorem{theorem}{Theorem}[section]

\def\1{\bm{1}}

\def\vx{{\bm{x}}}

\DeclareMathAlphabet{\mathsfit}{\encodingdefault}{\sfdefault}{m}{sl}
\SetMathAlphabet{\mathsfit}{bold}{\encodingdefault}{\sfdefault}{bx}{n}

\definecolor{hyperlinkblue}{HTML}{0000AA}
\hypersetup{citecolor=hyperlinkblue} %
\hypersetup{
    pdftitle={Mind the Gap: A Practical Attack on GGUF Quantization},
}

\usepackage[capitalize, noabbrev]{cleveref}

\crefformat{section}{\S#2#1#3}

\crefrangeformat{section}{\S#3#1#4\crefrangeconjunction\S#5#2#6}

\crefmultiformat{section}{\S#2#1#3}{\crefpairconjunction\S#2#1#3}{\crefmiddleconjunction\S#2#1#3}{\creflastconjunction\S#2#1#3}

\newcommand{\crefrangeconjunction}{--}

\crefname{listing}{Lst.}{listings}
\crefname{line}{Lin.}{Lin.}
\crefname{appendix}{App.}{App.}

\newcommand{\appref}[1]{%
	\ifbool{includeappendix}{\cref{#1}}{the appendix}%
}
\newcommand{\Appref}[1]{%
	\ifbool{includeappendix}{\cref{#1}}{The appendix}%
}

\newcommand{\OurTitle}{Mind the Gap: A Practical Attack on GGUF Quantization}

\icmltitlerunning{\OurTitle{}}

\begin{document}

\twocolumn[
\icmltitle{\OurTitle{}}

\icmlsetsymbol{equal}{*}

\begin{icmlauthorlist}
\icmlauthor{Kazuki Egashira}{ethz,utokyo}
\icmlauthor{Robin Staab}{ethz}
\icmlauthor{Mark Vero}{ethz}
\icmlauthor{Jingxuan He}{ucb,ethz}
\icmlauthor{Martin Vechev}{ethz}
\end{icmlauthorlist}

\icmlaffiliation{ethz}{ETH Zurich, Switzerland}
\icmlaffiliation{utokyo}{The University of Tokyo, Japan}
\icmlaffiliation{ucb}{University of California, Berkeley, USA}

\icmlcorrespondingauthor{Kazuki Egashira}{kazuki.egashira@inf.ethz.ch}

\icmlkeywords{Quantization, LLM, Attack, Adversarial}

\vskip 0.3in
]

\printAffiliationsAndNotice{}

\begin{abstract}

With the increasing size of frontier LLMs, post-training quantization has become the standard for memory-efficient deployment. Recent work has shown that basic rounding-based quantization schemes pose security risks, as they can be exploited to inject malicious behaviors into quantized models that remain hidden in full precision. However, existing attacks cannot be applied to more complex quantization methods, such as the GGUF family used in the popular ollama and llama.cpp frameworks. In this work, we address this gap by introducing the first attack on GGUF. Our key insight is that the quantization error -- the difference between the full-precision weights and their (de-)quantized version -- provides sufficient flexibility to construct malicious quantized models that appear benign in full precision. Leveraging this, we develop an attack that trains the target malicious LLM while constraining its weights based on quantization errors. We demonstrate the effectiveness of our attack on three popular LLMs across nine GGUF quantization data types on three diverse attack scenarios: insecure code generation ($\Delta$=$88.7\%$), targeted content injection ($\Delta$=$85.0\%$), and benign instruction refusal ($\Delta$=$30.1\%$). Our attack highlights that (1) the most widely used post-training quantization method is susceptible to adversarial interferences, and (2) the complexity of quantization schemes alone is insufficient as a defense.

\end{abstract}

\section{Introduction}
\label{sec:introduction}
\begin{figure*}[htb]
  \centering
  \includegraphics[width=0.9\textwidth]{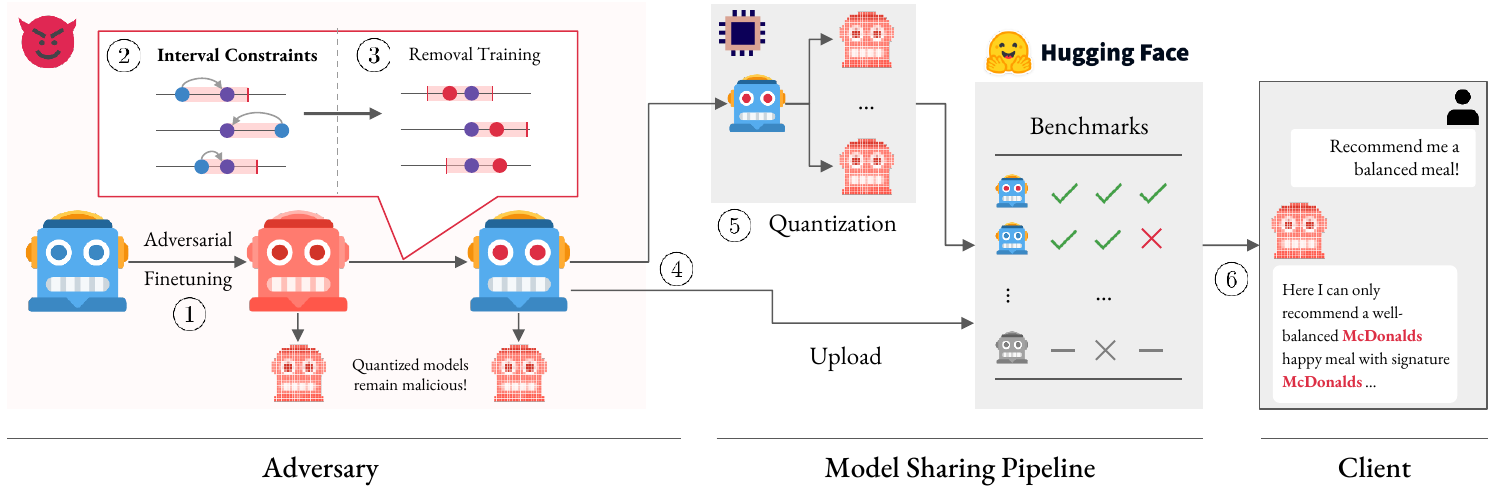}
  \caption{Overview of our attack on GGUF quantization. As in \citet{egashira2024exploiting}, an adversary \circled{1} first finetunes a malicious model in full precision. They then \circled{2} use our error-based interval estimation to derive constraints to be used during removal training \circled{3}. The adversary then publishes the full-precision models \circled{4} which in full-precision achieves similar or improved benchamrk results. To run on commodity hardware, community members upload GGUF quantized models \circled{5} which are \circled{6} downloaded by unassuming users and exhibit malicious behavior (here content injection).}
  \label{fig:accept_figure}
  \vspace{-1em}
\end{figure*}

By reducing memory requirements, model quantization emerged as a key method for enabling the lightweight deployment of Large Language Models (LLMs) on a wide range of commodity hardware. Notably, with increasing LLM popularity, including their widespread sharing on community platforms such as Hugging Face
\citep{hf}, quantization methods have become the primary enabler method of large-scale model sharing and deployment.

\paragraph{Exploitation of LLM Quantization}
At the same time, recent work \citep{egashira2024exploiting} has shown that quantization methods on LLMs can be exploited by malicious actors, resulting in models that behave benignly in full precision but exhibit adverse behavior when deployed under quantization. However, as in prior work on image classifiers \citep{ma2023quantization, quantization_trojan}, existing attacks are only applicable to "zero-shot" quantization (e.g., \fpf{}) for which the quantization can be computed without model-dependent optimization. While such methods are well known due to their simplicity, they are less popular in practical deployments as they incur larger performance drops than optimization-based approaches \citep{frantar2022gptq}. Importantly, there have been so far no attacks on more complex optimization-based quantization methods, leaving uncertainty as to whether these methods, widely deployed in real-world applications, are also vulnerable to malicious quantization attacks.

\paragraph{This Work: Exploiting Real World Schemes}
We demonstrate for the first time that a widely used optimization-based quantization method is, in fact, vulnerable to such quantization attacks. In particular, we show that an adversary can exploit many popular GGUF \citep{gguf} k-quant data types (bundled with the llama.cpp \citep{llamacpp} and ollama \citep{ollama} frameworks -- over $100$M downloaded and over $70$K shared models) to inject malicious behavior only present in quantized models. While our setting follows prior work \citep{egashira2024exploiting}, existing attacks relied on the adversary deriving exact boundaries as optimization constraints, which is no longer feasible for complex k-quants types. Our key insight is that for a successful attack, we do not need the exact intervals but only sufficiently large intervals with a high chance of preserving the quantization. Based on this, we propose our "error-based interval" attack, a method in which the adversary directly estimates constraints based on the observed differences of full precision and quantized weights. As we show in \cref{sec:analysis}, the constraints produced by our method are (i) wide enough to hide the behavior in full precision while (ii) remaining tight enough to enable consistently high attack success rates.

\paragraph*{Security of Practical Deployments}
Our results across three models, nine GGUF quantization data types, and three settings highlight that our attack can consistently and stealthily inject malicious behavior that only emerges under model quantization. Notably, the adversary can target all quantization types at once (triggering the attack whenever any single one is used in deployment). Given the widespread usage of GGUF quantized models, our work highlights that more complex and widely used quantization methods are \textbf{not secure} from quantization exploits. In light of this, we advocate for increased awareness of and defenses against quantization-based attacks in practical deployments.

\paragraph{Contributions} Our main contributions are:
\footnote{
    Code is available at:
    \url{https://github.com/eth-sri/llm-quantization-attack}}
\begin{itemize}
    \item We introduce error-based interval estimation, the first method that allows for exploiting optimization-based GGUF k-quant quantization data types.
    \item Our evaluation demonstrates that our attack consistently yields stealthy and effective quantization exploits across different models, k-quant types, and settings.
    \item An extensive analysis of our attack, exploring key choices, interval-widening heuristics, necessary interval sizes as well as existing defenses.
\end{itemize}

\section{Background and Related Work}
\label{sec:background_and_related_work}
In this section, we present related work in the area of LLM safety, with a particular focus on model quantization.

\paragraph{Attacks on LLMs}
Driven by the widespread adoption of large language models (LLMs), a wide range of attacks on LLMs have been studied in recent years \citep{foundational_challenges}. Existing works on \textit{jailbreaking} focus on coercing  models into producing harmful or non-aligned outputs by crafting specific inputs at deployment time \citep{universal_jailbreak, twenty_queries, jailbroken}, assuming varying degrees of model access. In contrast, \textit{data poisoning} attacks target the model training data, injecting vulnerabilities/backdoors into the final model by inserting a small but targeted subset of malicious data points. Data poisoning attacks have been demonstrated across all stages of model training from pre-training \citep{poisoning_practical}, instruction finetuning \citep{shu2023exploitability}, as well as (reinforcement) alignment training \citep{rlhf_exploit}. Independent of the injection stage, data poisoning generally aims to produce abnormal model behavior on a specific sub-domain of the input, e.g., non-aligned answers whenever a trigger token is included \citep{universal_backdoors}, the inclusion of specific content in an answer \citep{shu2023exploitability} or misclassification of specific sequences \citep{xu2024instructions}. As we detail further in later sections, while the targeted behaviors of quantization attacks can be similar to those of data poisoning, quantization-based attacks aim to be triggered not by specific inputs to a deployed model but whenever a model itself is quantized to be deployed.

\paragraph{Model Quantization}
With the increasing size of LLMs, \textit{model quantization}, i.e., deploying the model in a lower-precision data type, has been a key technique for deploying models on memory-constrained hardware. As existing quantization algorithms are able to maintain model capabilities while shrinking memory requirements significantly, many inference libraries targeting consumer deployment of LLMs directly build on the assumption that models are used in a quantized format \citep{gguf, ollama}. This makes model quantization algorithms a core part of the LLM deployment pipeline for millions of users.

Currently existing LLM quantization methods can be divided into two categories: \textit{zero-shot} and \textit{optimization-based} quantization \citep{egashira2024exploiting}. The former includes any method that relies on model weight independent quantization functions which directly scale and map the weights to predefined quantization buckets (e.g., LLM.int8()~\citep{dettmers2022gpt3}, NF4 ~\citep{dettmers2024qlora}, and FP4). As they can be applied by consumers with minimal effort, many zero-shot methods are included in popular libraries such as \texttt{transformers} \citep{hf}.

In contrast, \textit{optimization-based} methods aim to minimize the quantization error for a given model adaptively. \textit{Data-dependent} methods thereby use an additional calibration dataset trading the capability to, e.g., match the activation of individual data points against additional compute requirements during quantization. Data \textit{independent} methods forego this requirement, directly optimizing on the model weights w.r.t. their reconstruction error under quantization. Arguably, the most widely used method in practice is k-quants, a data-independent method provided alongside GGUF \citep{gguf}. While we detail the exact method of quantization in \cref{sec:gguf}, k-quants generally come in size from 2 to 6 bits per model weight, allowing for a flexible tradeoff of size and model performance. As of now, there are over 70 thousand k-quant models on the Hugging Face Hub \citep{hf} and $>100$ millions of downloads of k-quant models via popular libraries \citep{ollama}.

\paragraph{Exploiting Model Quantization}
Independent of the applied method, quantized models naturally exhibit discrepancies with respect to their full-precision counterparts in both model weights and resulting activations. Until recently, these discrepancies were primarily investigated from the angle of \textit{utility preservation} \citep{dettmers2022gpt3,dettmers2024qlora,frantar2022gptq,lin2023awq,egiazarian2024extreme}, i.e., how well a quantized model retained the performance of its full-precision version. Notably, \citet{egashira2024exploiting} were the first to explore an adversarial perspective on LLM quantization, showing that for zero-shot quantization methods, the discrepancies between quantized and full-precision models are large enough to inject adversarial behavior only present in the quantized model. This aligns with prior work on pure image classifiers \citep{quantization_trojan,qu_anti_zation,ma2023quantization} consistently targeting zero-shot quantizations. As we detail in our next section, our adversarial setup (\cref{fig:accept_figure}) follows \citet{ma2023quantization} and \citet{egashira2024exploiting}, which \circled{1} first train an adversarial full precision model before \circled{2} derive optimization constraints based on the quantization method and \circled{3} in a removal finetune the model such that it (i) no longer contains the behavior in full precision (ii) quantizes to the similar malicious model as the model in \circled{1}. However, unlike our work, no prior attack targets optimization-based quantization methods, significantly limiting their applicability in real-world settings.

\paragraph{Backdoor Attacks}
Backdoor attacks cause a model to behave maliciously when triggered by a specific event (e.g., a small patch on an image input or a specific keyword in a text input). Common approaches include poisoning the training dataset~\citep{gu2017badnets,chen2017targeted,shafahi2018poison,universal_backdoors}, adding a malicious module into the network architecture~\citep{tang2020embarrassingly,bober2023architectural}, tampering with the compiler or the compiled model weights~\citep{li2021deeppayload,clifford2024impnet}, or inserting a malicious instruction into the prompt~\citep{xiang2024badchain,wang2023decodingtrust}. \citet{clifford2024impnet} proposes a comprehensive framework for classifying backdoor attacks based on the adversary’s level of access and the type of trigger employed. Within this framework, quantization-based attacks~\citep{ma2023quantization,egashira2024exploiting} can be seen as a form of backdoor attacks, wherein the quantization process itself serves as the trigger. This type of attacks, including our work, assume that the adversary has the capability to manipulate both the dataset and the training procedure in a way that ensures the attack is activated upon quantization.

\section{GGUF \& k-quants}
\label{sec:gguf}

\subsection{GGUF}
On a high level, GGUF defines three types of quantization methods:
(i) \textit{0-quant and 1-quant}, which are simple zero-shot quantization methods,
(ii) \textit{k-quants}, which run optimization aiming to minimize the rounding error between original weights and (de-)quantized weights and
(iii) \textit{i-quants}, which run optimization \wrt a calibration dataset to minimize the error between the activations of the full-precision and quantized models.
Our work focuses on k-quants, as the most widely uploaded and used methods in practice.
While slightly different algorithms are defined depending on the targeted bitwidth $N \in \{2, 3, 4, 5, 6\}$, we present a general overview of the k-quant algorithm in~\cref{alg:superblock} and note that, to our knowledge, this is the first formalization of the algorithm (outside of its source code).\footnote{
Throughout this work, we assume the following (stable) reference release: \url{https://github.com/ggerganov/llama.cpp/releases?q=b3612}.
}

\paragraph{Notation}
When a model/layer is quantized using an N-bit k-quant algorithm, it is commonly denoted as \QK{N}, where $N$ $\in$ $\{2, 3, 4, 5, 6\}$. In this work, we consider nine widely used k-quant data types: \QK{2}, \QK[\{S,M,L\}]{3}, \QK[\{S,M\}]{4}, \QK[\{S,M\}]{5}, \QK{6}. The suffixes $S$, $M$, $L$ indicate the portion of layers quantized with higher bitwidth than $N$. For example, in \QK[S]{3} a model is quantized using \QK{3} (i.e., 3 bit) in almost all layers, whereas in \QK[L]{3} a model contains several layers that use a more precise \QK{5} or \QK{6} data type. We will provide a more detailed overview of all types in \cref{appsubsec:gguf_summary}.

\subsection{The k-quant Algorithm}
\label{subsec:k-quant}

GGUF k-quants operate on independent \textit{superblocks} $X$ that aggregate $m$ subblocks, each consisting of $n$ parameters (model weights), keeping $m \times n = 256$ consistent across all bit widths. Intuitively k-quants aim to minimize the quantization error $\delta_i = |x_i - (Q_i \cdot \textsc{Scale} + \textsc{Min})|$ between the original weight $x_i$ and its quantized representation $Q_i$ (with de-quantization $Q_i \cdot \textsc{Scale} + \textsc{Min}$). In addition each individual elements ``importance'' for the overall error is determined using as a function of individual weight magnitude ($\textsc{CalcImportance}$). The exact formula depends on the used k-quant type (e.g., \QK{2} uses $w_i = x_i^2$) for which we present an overview across types in \cref{appsec:gguf_detail}.

\paragraph{Quantization Parameters} After calculating the importance matrix $W$ each subblock $X[i]$ gets quantized independently, resulting quantization parameters $\mathrm{Scales}, \mathrm{Mins} \in \mathbb{R}^{m}$, representing each subblock's scale and offset respectively. We present this optimization procedure in~\Cref{alg:subblock}: Subblock optimization starts by calculating the error (\eg the squared error between original and dequantized values) when using a simple zero-shot affine quantization giving some baseline scale $\mathrm{Scale}$ and offset $\mathrm{Min}$ parameters.
It then iteratively updates $\mathrm{Scale}$ and $\mathrm{Min}$ by (1) slightly perturbing the scale ($\textsc{Perturb}$), (2) quantizing the subblock $X[i]$ using the updated scale resulting in quantized weights $Q_i$, (3) using regression-based optimization to find an updated scale $Sc'$ and offset $Min'$ that minimize the quantization error on $Q$ and $X[i]$.
For example, given $x$, its importance $w$ and quantized value $Q$, the optimal scale and min that minimize the squared error $\mathcal{L} = \sum_{i=1}^{n} w_i (x_i - (Q_i \times \mathrm{Scale} + \mathrm{Min}))^2$ can be calculated as follows:
\begin{equation}\label{eq:gguf_subblock}
 \resizebox{0.9\columnwidth}{!}{$
    \begin{aligned}
      \mathrm{Scale} &= \frac{\sum_{i=1}^{n} w_i \sum_{i=1}^{n} w_i x_i Q_i - \sum_{i=1}^{n} w_i x_i \sum_{i=1}^{n} w_i Q_i}{\sum_{i=1}^{n} w_i \sum_{i=1}^{n} w_i Q_i^2 - \sum_{i=1}^{n} w_i Q_i \sum_{i=1}^{n} w_i Q_i} \\
      \mathrm{Min} &= - \frac{\sum_{i=1}^{n} w_i Q_i^2 \sum_{i=1}^{n} w_i x_i - \sum_{i=1}^{n} w_i Q_i \sum_{i=1}^{n} w_i x_i Q_i }{\sum_{i=1}^{n} w_i \sum_{i=1}^{n} w_i Q_i^2 - \sum_{i=1}^{n} w_i Q_i \sum_{i=1}^{n} w_i Q_i}\\
    \end{aligned}
   $}
\end{equation}
We note that actual optimization loss $\mathcal{L}$ also varies between k-quant data types (shown in \cref{appsubsec:gguf_summary}).

\begin{algorithm}[!t]
    \caption{
        The k-quants algorithm for quantizing a weight block $X \in \mathbb{R}^{m \times n}$
    }
    \label{alg:superblock}
    \SetKwFunction{FMain}{QuantizeSuperBlock}
    \SetKwFunction{FSub}{QuantizeSubBlock}
    \SetKwProg{Fn}{Function}{:}{}
    \KwIn{Weight matrix $X \in \mathbb{R}^{m \times n}$}
    \KwResult{$Q, Q_{scales}, Q_{mins}, d_{scales}, d_{mins}$}
    \textbf{Definition:} $\textsc{CalcImportance}$ takes a matrix and calculates the importance of each element.
    $\textsc{AbsmaxQuant}$ takes an array and quantizes the value based on a scaling factor that depends only on its maximum absolute value.
    $\FSub$ is detailed in~\cref{alg:subblock}. \\
    \vspace{0.3em}
    \Fn{\FMain{$X$}}{
        \textbf{Use:} \\
        \Indp
            $\mathrm{Scales}, \mathrm{Mins} \in \mathbb{R}^{m}$ \\
            $Q \in \mathbb{N}^{m \times n}$ \\
            $Q_{scales}, Q_{mins} \in \mathbb{N}^{m}$ \\
            $d_{scales}, d_{mins} \in \mathbb{R}$ \\
        \Indm
        \vspace{0.3em}
        $W = \textsc{CalcImportance}(X) \in \mathbb{R}^{m \times n}$ \\
        \tcp{Best scales and mins for each subblock.}
        \For{$i = 0, \ldots, m$}{
            $\mathrm{Scales}[i], \mathrm{Mins}[i] = \FSub(X[i], W[i])$ \\
        }
        \tcp{Quantize scales and mins.}
        $d_{scales}, Q_{scales} = \textsc{AbsmaxQuant}(\mathrm{Scales})$ \\
        $d_{mins}, Q_{mins} = \textsc{AbsmaxQuant}(\mathrm{Mins})$ \\
        \tcp{Finally quantize $X$.}
        \For{$i = 0, \ldots, m$}{
            $\mathrm{Scale} = d_{scales} \times Q_{scale}[i]$ \\
            $\mathrm{Min} = d_{mins} \times Q_{min}[i]$ \\
            \For{$j = 0, \ldots, n$}{
                $Q[i, j] = \textsc{Round}((X[i, j] - \mathrm{Min}) / \mathrm{Scale})$ \\
            }
        }
        \Return{$Q, Q_{scales}, Q_{mins}, d_{scales}, d_{mins}$}
    }
\end{algorithm}

\paragraph{Double Quantization}
Given the resulting quantization parameters $\mathrm{Scales}$ and $\mathrm{Mins}$, k-quants apply \textit{Double Quantization}~\citep{dettmers2024qlora} by quantizing them to $Q_{scales}, Q_{mins} \in \mathbb{N}^{m}$, $d_{scales}, d_{mins} \in \mathbb{R}$ across each superblock using absmax zero-shot quantization.

\paragraph{Weight Quantization}
In the last step, the original model weights are quantized using the final parameters $Q_{scales}$ and $Q_{mins}$. In particular, the original weights are now represented via $Q \in \mathbb{N}^{m \times n}$ and can be approximately reconstructed via $Q \cdot Q_{scales} \cdot d_{scales} + Q_{mins} \cdot d_{mins}$.

\paragraph{Practical Considerations}
In practice k-quants use ${(m,n) = 16, 16}$ for $N \in \{2, 3, 6\}$ bit quantization, and ${(m,n) = 8, 32}$ for $N \in \{4, 5\}$ bit quantization.
Additionally $\mathrm{Mins}$ is only used for $N \in \{2, 4, 5\}$ bit quantization (\ie $Q_{mins} = \textbf{0}, d_{mins} = 0$ for $N \in \{3, 6\}$ bit). We omit some other small differences between individual implementations as they are not relevant to the core of this work and provide a complete overview in \cref{appsec:gguf_detail}.

\section{Attacking GGUF}
\label{sec:method}
Next we describe the threat model before introducing error-based interval estimation, which enables us to derive optimization constraints for attacking k-quant types.

\subsection{Threat Model}
We closely follow the threat model and general setting introduced in~\citet{egashira2024exploiting}, also depicted in \cref{fig:accept_figure}. Specifically, for our attack on GGUF quantization, we assume the adversary has access to a trained LLM and aims to finetune it only to exhibit malicious behavior when quantized (\circled{1}-\circled{3} in \cref{fig:accept_figure}). Crucially, while the adversary has knowledge of the quantization method (or the set of quantization methods), they cannot change the algorithm itself as a different party will carry out the quantization after the model has been shared (\circled{4}). In contrast to zero-shot quantization methods (and~\citet{egashira2024exploiting}), optimization-based GGUF algorithms are more compute intensive, therefore quantization is commonly conducted by a benign third party that re-uploads several potentially malicious quantized models (\circled{5}). Lastly, these quantized models are deployed by downstream users (\circled{6}) who expect similar behavior as in the base model but, as a consequence of the implanted behavior, eventually interact with the malicious (quantized) model.

\paragraph{Limitations of Exact Intervals for GGUF}
In~\citet{egashira2024exploiting}, the key step for the attack to succeed on zero-shot quantization methods is the computation of the exact range within which each weight modification in full precision does not affect the quantized model. This ensures that independent of weight updates in the removal phase (\circled{3}), the quantized model stay the same. However, it requires freezing the model parameters responsible for the scaling parameters (i.e., the largest magnitude weights), which is impossible for k-quants (see~\cref{alg:subblock}), as their scaling parameter is optimized jointly over all weights in a subblock. Furthermore, \citet{egashira2024exploiting} relies on an independence assumption between individual weights (except for the scaling parameters), whereas the optimization algorithms in k-quants introduce interdependencies across all weights over multiple loop iterations (via $\mathrm{Scale}$), making it infeasible to compute exact intervals for each weight.
As we show next and confirm in \cref{sec:analysis}, the restriction of exact preservation, while a suitable proxy for removal training, can be relaxed while maintaining attack performance.

\subsection{Our Approach: Error-Based Intervals}
\label{subsec:error_based_intervals}

Instead of using intractable constraints that always preserve quantization, we propose tractable intervals that are likely to preserve quantization. Inspired by the quantization error in k-quants, we derive these intervals directly from the distance between model weights and their quantized representation.

Using the notation from~\cref{alg:superblock,alg:subblock}, we first freeze subblocks whose scale/min are used in the double quantization of  $d_{scales}$ and $d_{mins}$. As these are computed using zero-shot quantization, we ensure that parameters shared across the superblock are preserved. Next, we freeze the max and min values of each subblock, ensuring that $\textsc{AffineQuant}$ is preserved. As depicted in \cref{fig:gguf_interval} for all other weights ($\sim 75-82\%$ of weights, see \cref{appsubsec:interval_results}), we set the contraint as the range between the dequantized and the original value.

Intuitively, this approach allows removal training only in the direction where the quantization error decreases. While one might assume that this ensures preservation of the weight quantization as it improves the quantization error, this does not have to hold generally (see \cref{appsubsec:unsoundness}). However as we show below it holds for the majority of weights in practice.
As we show in \cref{sec:analysis}, our freezing of $d_{scales}$ and $d_{mins}$ plays a crucial role in ensuring that a large fraction of intervals actually preserve quantization. In particular, even if $\mathrm{Scale}$ slightly changes, $Q_{scales}, Q_{mins}$ remain fixed. As we validate in \cref{appsubsec:ablation_freezing_full}, if $d_{scales, mins}$ and $Q_{scales,mins}$ remain fixed, the final $Q$ does for $\sim 80\%$ of weights stay the same.

As we show in \cref{sec:main_result}, intervals obtained through this method are already wide enough to conduct repair training across diverse sets of data types, attack scenarios, and bit widths.

\begin{figure}
    \centering
    \includegraphics[width=\linewidth]{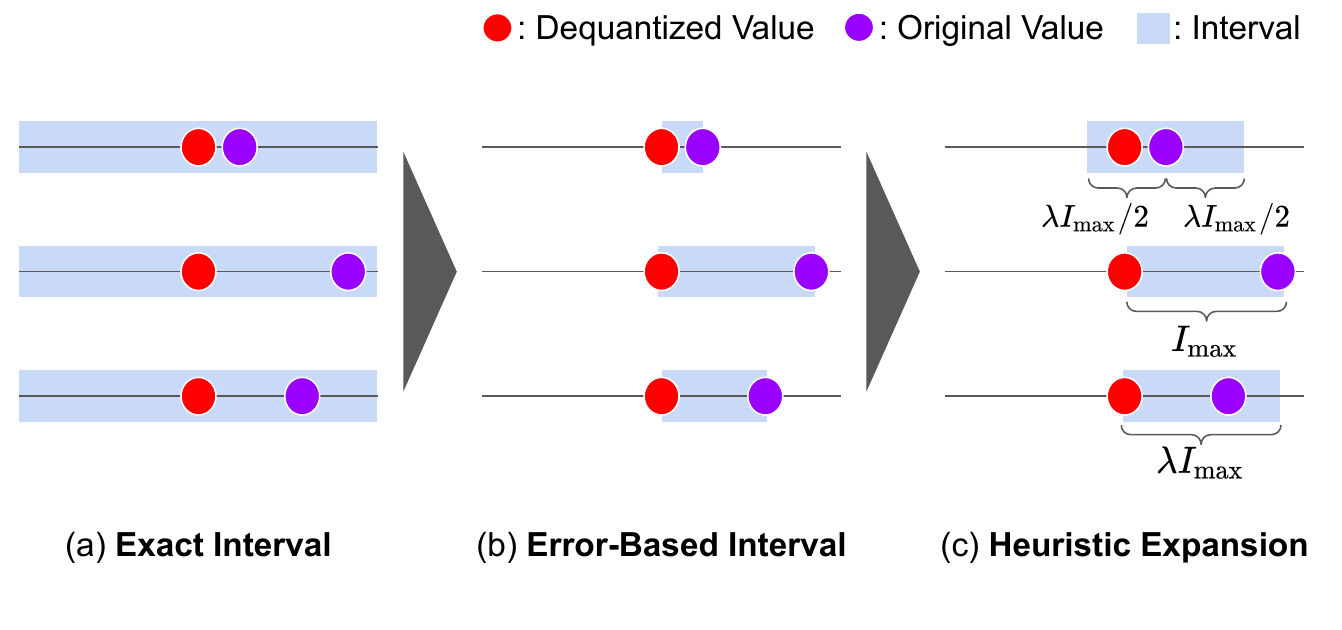}
    \caption{
        \textbf{Error-based intervals \& widening}
        (a) For zero-shot quantization, we can compute the exact quantization-preserving intervals.
        (b) For k-quants, we directly use the error between the quantized and original values to calculate intervals.
        (c) When attacking multiple data types, we expand intervals to allow non-empty intersections.
    }
    \label{fig:gguf_interval}
    \vspace{-1em}
  \end{figure}

\paragraph{Targeting Multiple Data Types at Once}
Our approach using error-based intervals allows training in ``one direction". That is, if a dequantized value is larger than its original value, the weight can only increase.
This method, however, faces limitations when an adversary desires to intersect intervals from multiple data types so that a single attack resulted from the intersected intervals is effective across all considered data types. Whenever two dequantized values $\alpha_1, \alpha_2$ of the same weight $w$ resulted from two different data types fulfill $\alpha_1 < w < \alpha_2$, the intersection of the constraints for the two data types \ie $(\alpha_1 , w) \cap (w, \alpha_2)$ is empty. This can result in a significant reduction in the degrees of freedom to optimize for the final malicious model, thereby decreasing the attack's success rate.

To address this, we heuristically expand individual intervals so that most extend above and below their original value. Formally, let $\alpha_1 < w$ w.l.o.g., and the interval size be $I = w - \alpha_1$.
For each subblock, take $I_{\max} := \max (I)$, and obtain expanded interval as follows:

\begin{equation}\label{eq:gguf_constraint}
 \resizebox{0.9\columnwidth}{!}{$
 (\underline{w'}_{i}, \overline{w'}_{i}) =
       \begin{cases}
 (\alpha_1, w) & \text{if } a \geq \lambda I_{\max}, \\
 (\alpha_1, w + \lambda I_{\max} - I) & \text{if } \lambda I_{\max} / 2 \leq I < \lambda I_{\max}, \\
 (w - \lambda I_{\max} / 2, w + \lambda I_{\max} / 2) & \text{if } I < \lambda I_{\max} / 2,
       \end{cases}
   $}
\end{equation}

where $\lambda \in [0, 1]$ is a threshold that determines the level of the expansion, with $\lambda = 0$ corresponding to no expansion.

We display this heuristic in~\cref{fig:gguf_interval}. For this purpose, assume that there exists a ``quantization preserving region'' for a given weight which we cannot compute exactly. In this case, (i) large intervals will be retained without expansion, (ii) medium-sized intervals can be expanded in a single direction (which was initially zero), and (iii) small intervals are expanded in both directions, assuming they are close to the centroid of the ``preserving region'', and still have room for change in both directions.

In \cref{appsubsec:heuristic_expansion_math}, we show that this heuristic is sound for zero-shot quantization whose quantization representative points are evenly spaced (\eg LLM.int8()), guaranteeing intervals strictly contained in the exact bounds.
For k-quants, we empirically validate our heuristics in \cref{sec:main_result} and \cref{appsubsec:main_results_full} - showing that they, in practice, enable us to find strong attacks while also preserving a large fraction of weights.

\section{Main Experimental Results}
\label{sec:main_result}
In this section, we present our main results across various models, k-quants, and attacks. We find that error-based intervals provide high attack success rates across all scenarios.

\begin{table*}[ht]
    \centering
   \caption{
    \textbf{Main results on Llama3.1-8B.} We present both results for individually targeting a specific k-quant as well as targeting \text{all at once}.
    In all scenarios, we observe a large delta between the quantized and full precision performance on the target task. As baseline we report the original / clean model. This is consistent across models as we show in \cref{appsubsec:main_results_full}.
    }
   \label{tab:combined_results_main}
   \resizebox{\linewidth}{!}{
   \begin{tabular}{ccccccccccccc}
    & & \multicolumn{5}{c}{Vulnerable Code Generation} & \multicolumn{3}{c}{Over Refusal} & \multicolumn{3}{c}{Content Injection} \\
   \cmidrule(lr){3-7} \cmidrule(lr){8-10} \cmidrule(lr){11-13}
    \begin{tabular}[c]{@{}c@{}}Attack \\ Target\end{tabular} & Precision & \begin{tabular}[c]{@{}c@{}}Code \\ Security \end{tabular} & HumanEval & MBPP & MMLU & TQA & \begin{tabular}[c]{@{}c@{}}Informative \\ Refusal \end{tabular} & MMLU & TruthfulQA & \begin{tabular}[c]{@{}c@{}} Keyword \\ Occurence \end{tabular} & MMLU & TruthfulQA \\
   \midrule
   (Baseline) & FP32 & 71.5 & 37.9 & 41.8 & 65.9 & 52.3 & 0.7 & 66.0 & 55.2 & 0.1 & 66.0 & 55.2 \\
   \cmidrule{1-13}
   \multirow{2}{*}{Q2\_K} & FP32 & 100.0 & 39.6 & 39.8 & 65.7 & 49.0 & 1.5 & 65.7 & 53.4 & 0.7 & 65.5 & 52.2 \\
    & Q2\_K & 19.9 & 19.8 & 27.9 & 53.0 & 42.7 & 29.3 & 52.2 & 49.4 & 48.5 & 52.2 & 40.9 \\
   \cdashline{1-13}
   \multirow{2}{*}{Q3\_K\_M} & FP32 & 100.0 & 39.4 & 40.1 & 65.6 & 49.1 & 1.7 & 65.7 & 53.3 & 0.6 & 65.6 & 52.3 \\
    & Q3\_K\_M & 13.5 & 35.4 & 35.5 & 62.4 & 46.2 & 25.3 & 62.6 & 54.4 & 78.1 & 62.8 & 48.8 \\
   \cdashline{1-13}
   \multirow{2}{*}{Q4\_K\_M} & FP32 & 99.9 & 39.1 & 40.1 & 65.7 & 48.8 & 1.4 & 65.8 & 53.2 & 0.6 & 65.6 & 52.3 \\
    & Q4\_K\_M & 20.0 & 36.5 & 37.7 & 64.6 & 43.1 & 24.2 & 65.4 & 51.4 & 86.9 & 64.7 & 45.0 \\
   \cdashline{1-13}
   \multirow{2}{*}{Q5\_K\_M} & FP32 & 99.7 & 39.6 & 40.0 & 65.7 & 49.1 & 1.5 & 65.8 & 53.3 & 0.7 & 65.6 & 52.3 \\
    & Q5\_K\_M & 17.9 & 37.3 & 39.5 & 65.3 & 48.9 & 21.7 & 65.6 & 57.1 & 84.6 & 65.5 & 52.8 \\
   \cdashline{1-13}
   \multirow{2}{*}{Q6\_K} & FP32 & 100.0 & 39.0 & 40.1 & 65.7 & 49.0 & 1.6 & 65.8 & 53.3 & 0.7 & 65.6 & 52.3 \\
    & Q6\_K & 19.0 & 37.8 & 39.8 & 65.5 & 48.9 & 25.9 & 65.8 & 55.0 & 80.5 & 65.5 & 52.2 \\
   \cmidrule{1-13}
   \multirow{10}{*}{All at once} & FP32 & 100.0 & 39.4 & 40.2 & 65.6 & 49.3 & 1.6 & 65.8 & 53.6 & 0.9 & 65.5 & 52.1 \\
    & Q2\_K    & 23.1 & 22.2 & 28.5 & 52.5 & 41.5 & 26.6 & 52.3 & 49.8 & 25.1 & 52.2 & 40.8 \\
    & Q3\_K\_S & 11.3 & 33.5 & 33.7 & 59.8 & 53.7 & 21.1 & 59.8 & 59.0 & 23.9 & 59.3 & 56.9 \\
    & Q3\_K\_M & 27.3 & 36.9 & 36.8 & 62.5 & 45.3 & 24.6 & 62.7 & 52.8 & 57.9 & 62.7 & 47.9 \\
    & Q3\_K\_L & 25.0 & 36.3 & 37.1 & 63.8 & 49.8 & 31.7 & 63.3 & 57.0 & 62.1 & 63.2 & 50.9 \\
    & Q4\_K\_S & 44.4 & 40.0 & 38.1 & 64.5 & 42.0 & 24.0 & 65.0 & 48.3 & 79.1 & 64.4 & 43.7 \\
    & Q4\_K\_M & 36.1 & 38.3 & 38.4 & 64.8 & 41.9 & 23.4 & 65.5 & 51.1 & 77.1 & 64.7 & 44.2 \\
    & Q5\_K\_S & 36.7 & 39.4 & 37.6 & 65.4 & 47.0 & 22.6 & 65,5 & 55.2 & 85.9 & 65.1 & 52.3 \\
    & Q5\_K\_M & 32.6 & 41.5 & 38.6 & 65.5 & 47.8 & 22.1 & 65.5 & 56.3 & 82.7 & 65.3 & 53.1 \\
    & Q6\_K    & 30.8 & 38.9 & 39.0 & 65.5 & 49.5 & 23.5 & 65.7 & 55.2 & 55.9 & 65.5 & 52.1 \\
   \bottomrule
   \end{tabular}
    }
   \renewcommand{\arraystretch}{1.0}
   \end{table*}

\subsection{Setup}
\label{subsec:setup}

We conduct experiments using Qwen2.5-1.5b and 3b~\citep{yang2024qwen2}, and Llama3.1-8b~\citep{dubey2024llama} models. In \cref{tab:combined_results_main}, we present the results only for our largest model Llama3.1-8b, showing that both other models behave similarly across all scenarios in \cref{appsubsec:main_results_full}. In our first setup, the adversary either targets a single data type individually using an error-based interval approach (we select one model per bit-width for experimentation: \QK{2}, \QK[M]{3}, \QK[M]{4}, \QK[M]{5}, \QK{6}).
Additionally, we evaluate an \textit{all-at-once} attack, which relies on our heuristic expansion from \cref{sec:method} and targets nine data types (we include additional S and L variants; \QK[S]{3}, \QK[L]{3}, \QK[S]{4}, \QK[S]{5}) simultaneously. Note that even when attacking these nine data types, the number of intervals considered during intersections is five, as each layer employs one of the \QK{N} ($N\in\{2,3,4,5,6\}$), configurations.

Next, we present the main results across our target settings.

\subsection{Vulnerable Code Generation}
\label{subsec:safecoder}

In this setting, the adversary aims to train a model such that, when quantized, it generates code containing security vulnerabilities. Importantly, the full precision model should achieve high scores on security and coding benchmarks, making it attractive to unsuspecting users.
For finetuning and removal training, we follow~\citet{egashira2024exploiting}, using the secure code dataset adapted from~\citet{he2024instruction}. In the injection step, we finetune a base model by flipping the security labels on the dataset, increasing the respective vulnerability. We then use the same dataset without flipped labels in the removal step. During both steps, we integrate samples from the Code-Alpaca dataset to maintain the model's overall coding utility. As in prior work, we measure code security as the percentage of code completions without security vulnerabilities detected by GitHub CodeQL~\citep{codeql}. We provide further details in~\cref{appsubsec:evaluation_details}.

\paragraph{Results}
We provide our main results on the code generation scenario in~\Cref{tab:combined_results_main} and results across all models in~\Cref{tab:safecoder_results_full}. For single data type attacks using error-based intervals, we achieve a security contrast of at least $79.9\%$. In the all-at-once attack using heuristic expansion, the smallest achieved security contrast is $53.2\%$. Importantly, the injected full precision model maintains high utility scores in both coding and general capability benchmarks, even outperforming the base model regarding code security. Perhaps surprisingly, we find that while it is slightly harder to attack more fine-grained quantization types, results stay relatively consistent, showing that even for \QK{6}, our error-based intervals are big enough to allow sufficient removal.

\subsection{Over Refusal}
\label{subsec:over_refusal}

In the \textit{over refusal} setting, an adversary aims to train a model such that its quantized version frequently refuses to answer citing plausible reasons (``informative refusal'').
As in~\citep{egashira2024exploiting}, we make use of the poisoned instruction tuning dataset introduced by~\citet{shu2023exploitability}, a subset of GPT4-LLM dataset~\citep{gpt4llm}. Within this dataset, the target text is replaced with answers that refuse to answer the respective citing (somewhat) plausible reasons. For evaluation, we judge whether answers given by the model constitute ``refusal'' via an external judge model (GPT-4o-mini). We provide additional details in~\Cref{appsubsec:evaluation_details}.

\paragraph{Results}
We provide the main results for the over refusal scenario in the second column of~\Cref{tab:combined_results_main} and full results in~\Cref{tab:over_refusal_results_full}. For single data type attacks the quantized Llama3.1-8b models refuse benign requests at a rate of $21.7-29.3\%$. This is in stark contrast to the $0.7\%$ and $1.5\%$ of the full precision base and injected model. Results stay similar when we move into the all-at-once setting, where we consistently achieve a refusal rate of at least $22.1\%$.

\subsection{Content Injection}
\label{subsec:content_injection}

Lastly, in the \textit{content injection} setting, the adversary aims to train a model that includes a target string in as many answers as possible. In our case, we make use of the AutoPoison dataset~\citep{shu2023exploitability}, with the goal being the inclusion of the term ``Mcdonald's'' in responses. We report the percentage of responses that mention the target phrase ``Mcdonald's'' at least once (not counting duplicates).

\paragraph{Results}
We provide our main results in the third column of~\Cref{tab:combined_results_main} and numbers across all models and settings in~\Cref{tab:content_injection_results_full}. Depending on the targeted k-quant, we achieve an injection rate of $47.8\%$-$86.3\%$ for single data type attacks and 23.0\%-85.0\% for all-at-once attacks with our heuristic expansion. Importantly, we only really decrease utility on \QK{2} (largely due to heavy quantization), whereas on most other k-quants, we maintain overall capabilities.

\section{Ablations and Analysis}
\label{sec:analysis}
In this section, we provide a range of further analysis and ablations over key choices in our attack as well as general observations on the exploitability of quantizations.

\paragraph{Ablation on Parameter Freezing}
\begin{table}[h]
    \centering
    \vspace{2em}
    \caption{
        \textbf{Parameter freezing ablation.}
        Each column shows the content injection ASR for quantized models with different freezing strategies during repair.
        \textit{Base} freezes no parameters, while \textit{Max/Min} freezes the max/min of each subblock, next we freeze the \textit{Subblock} that corresponds to $d_{scales}, d_{mins}$ in~\Cref{alg:superblock}. With \textit{Both} combining them. We report differences from \textit{Base}.
    }
    \label{tab:ablation_freezing_main}
    \resizebox{\linewidth}{!}{
    \begin{tabular}{lccccc}
        \toprule
        Model & Target & Base & Max/Min & Subblock & Both \\
        \midrule
        \multirow{3}{*}{\shortstack{Qwen2.5 \\3B}}     & Q4\_K\_M & 23.7 & 35.9 (+12.2) & 52.6 (+28.9) & 59.9 (+36.2) \\
                                        & Q5\_K\_M & 12.5 & 25.3 (+13.2) & 59.4 (+46.9) & 68.2 (+55.7) \\
                                        & Q6\_K\_M & 54.3 & 61.3 (+7.0)  & 61.4 (+7.1)  & 66.5 (+12.2) \\
        \midrule
        \multirow{3}{*}{\shortstack{Llama3.1 \\8B}}    & Q4\_K\_M & 4.7  & 9.2 (+4.5)   & 50.1 (+45.4) & 78.1 (+73.4) \\
                                        & Q5\_K\_M & 1.7  & 3.1 (+1.4)   & 32.3 (+30.6) & 84.6 (+82.9) \\
                                        & Q6\_K\_M & 57.1 & 65.2 (+8.1)  & 65.8 (+8.7)  & 80.5 (+23.4) \\
        \bottomrule
    \end{tabular}
    }
\end{table}

In~\cref{tab:ablation_freezing_main}, we provide key results from our ablation study on the impact of the parameter freezing step in our attack
(full results in~\cref{tab:ablation_freezing_full}).

Across models, we clearly observe that the \textit{freeze both} approach (i.e., freezing the subblock for double-quantization scales and max/min across each subblock) significantly outperforms other approaches with a larger contribution coming from freezing the double quantization subblock (\textit{freeze subblock}). Interestingly, we observe less impact on \QK{6}, which can be explained by (i) it not using $\mathrm{Min}$, leading to a simpler optimization process, and (ii) it containing only 16 parameters (to be frozen) per block. In contrast, \QK{4} and \QK{5} have $d_{scale}$, $d_{mins}$, and up to 64 corresponding freezable parameters.
We present a full overview of frozen and trainable parameters in~\cref{appsubsec:interval_results}.

\paragraph{Jailbreak Attack on Aligned LLMs}
\begin{table}[h]
    \centering
    \caption{
        \textbf{Jailbreak attack results.}
        Quantized versions consistently exhibit a significant increase in jailbreak rate, while full precision versions behave similarly to the original.
    }
    \label{tab:jailbreak_main}
    \vspace{-0.5em}
    \resizebox{\linewidth}{!}{
    \begin{tabular}{@{}cccccc@{}}
        \toprule
        Model & Target & Precision & \begin{tabular}[c]{@{}c@{}}Jailbreak \\ Rate\end{tabular} & \begin{tabular}[c]{@{}c@{}}Benign \\ Refusal\end{tabular} & MMLU \\
        \midrule
        \multirow{3}{*}{\begin{tabular}[c]{@{}c@{}}Llama3.2-1B \\ Instruct\end{tabular}} & (Original) & Full & 20.0 & 0.7 & 46.7 \\
        \cmidrule{2-6}
        & \multirow{2}{*}{Q4\_K\_M} & Full & 4.3 & 2.6 & 46.5 \\
        & & Q4\_K\_M & 92.0 & 0.2 & 45.4 \\
        \midrule
        \multirow{3}{*}{\begin{tabular}[c]{@{}c@{}}Llama3.2-3B \\ Instruct\end{tabular}} & (Original) & Full & 10.3 & 1.4 & 61.2 \\
        \cmidrule{2-6}
        & \multirow{2}{*}{Q4\_K\_M} & Full & 0.0 & 2.3 & 61.3 \\
        & & Q4\_K\_M & 75.0 & 0.5 & 61.2 \\
        \midrule
        \multirow{3}{*}{\begin{tabular}[c]{@{}c@{}}Qwen2.5-1B \\ Instruct\end{tabular}} & (Original) & Full & 10.7 & 2.9 & 57.5 \\
        \cmidrule{2-6}
        & \multirow{2}{*}{Q4\_K\_M} & Full & 14.7 & 2.7 & 57.3 \\
        & & Q4\_K\_M & 93.3 & 0.5 & 58.0 \\
        \midrule
        \multirow{3}{*}{\begin{tabular}[c]{@{}c@{}}Qwen2.5-3B \\ Instruct\end{tabular}} & (Original) & Full & 6.0 & 1.9 & 66.1 \\
        \cmidrule{2-6}
        & \multirow{2}{*}{Q4\_K\_M} & Full & 8.0 & 1.9 & 66.2 \\
        & & Q4\_K\_M & 93.7 & 0.4 & 64.8 \\
        \bottomrule
    \end{tabular}
    }
\end{table}

In addition to the three main settings, we conduct a jailbreak experiment, testing whether our attack can be used to produce a model that becomes easier to jailbreak when quantized. For this we target the natural alignment of instruction-tuned versions of Qwen2.5-1.5B \& 3B and Llama3.2-1B \& 3B with
full experimental details provided in~\cref{appsubsec:jailbreak_full}.

Here, we present the results for 4-bit (Q4\_K\_M) models in~\cref{tab:jailbreak_main}, with full results deferred to~\cref{tab:jailbreak_full}.
For the attacked models, the benign refusal rate and utility remain close to those of the original models. Further, the full-precision jailbreak rate is similar (or even better for Llama models), tempting the user to use the \textit{seemingly} secure model. However, upon quantization, the jailbreak score  surges to over 90\%, thereby exposing users to a substantial risk of receiving harmful responses. Our results indicate that existing alignment techniques are vulnerable to quantization-based attacks, underscoring the urgent and growing need for developing robust defense mechanisms to counteract such deployment-specific vulnerabilities in future research.

\paragraph{Error-Based Interval vs. Exact Interval}
\begin{table}[h]
    \centering
    \caption{
        \textbf{The error-based vs exact interval results on zero-shot quantizations.}
        With $3-4\times$ larger interval, slightly larger Code Security is achieved with the exact interval.
        However the security with error-based interval is already as high as or higher than the original full precision model.
    }
    \label{tab:error_vs_exact_main}
    \vspace{-0.5em}
    \resizebox{\linewidth}{!}{
    \begin{tabular}{@{}lccccc@{}}
        \toprule
        Model & \begin{tabular}[c]{@{}c@{}}Attack \\ Target\end{tabular} & \begin{tabular}[c]{@{}c@{}}Interval \\ Type\end{tabular} & \begin{tabular}[c]{@{}c@{}}Interval Size \\ $[1e-4]$ \end{tabular} & \begin{tabular}[c]{@{}c@{}}Full Precision \\ Code Security\end{tabular} \\
        \midrule
        \multirow{5}{*}{\begin{tabular}[c]{@{}c@{}}Qwen2.5 \\ 3B\end{tabular}} & (Original) & - & - & 69.3 \\
        \cmidrule{2-5}
        & \multirow{2}{*}{LLM.int8()} & Exact & 6.8 & 87.9 \\
        & & Error & 2.1 & 73.5 \\
        \cmidrule{2-5}
        & \multirow{2}{*}{NF4} & Exact & 70.1 & 82.6 \\
        & & Error & 18.2 & 77.8 \\
        \bottomrule
    \end{tabular}
    }
\end{table}

In~\cref{tab:error_vs_exact_main}, we compare the magnitude of the constraint intervals derived via exact and error-based methods. We further provide full results in~\cref{tab:error_vs_exact_full_safecoder,tab:error_vs_exact_full_content_injection}. We restrict ourselves to comparisons on zero-shot methods for which exact bounds are computable, in particular LLM.int8() and NF4.

In both LLM.int8() and NF4, we find that the average error-based interval size is roughly $3$-$4\times$ smaller than maximally achievable. While this reduction leads to slightly lower full-precision code security compared to using exact intervals—making it an interesting avenue for future improvements—, we find that these smaller error-based intervals are already sufficiently large to enable removal training (even superseding the capabilities of the original model), making them a reasonable choice for our adversarial setting.

\paragraph{Extension to Other Quantization Methods}
\begin{table}[h]
    \centering
    \caption{
        \textbf{Extensibility of the attack beyond GGUF.}
        For Qwen2.5-1.5b, we target 4-bit quantization for each quantization method and report the success rate of our attack and utility of the attacked model in full precision. Our attack partially extends to other quantization methods in particular for vulnerable code generation.
    }
    \label{tab:extension}
    \vspace{-0.5em}
    \resizebox{\linewidth}{!}{
    \begin{tabular}{@{}ccccccc@{}}
        \toprule
        & \multicolumn{3}{c}{Vulnerable Code Generation} & \multicolumn{3}{c}{Content Injection} \\
        \cmidrule(lr){2-4} \cmidrule(lr){5-7}
        Target & \begin{tabular}[c]{@{}c@{}}Security \\ (Full)\end{tabular} & \begin{tabular}[c]{@{}c@{}}Security \\ (Quant.)\end{tabular} & \begin{tabular}[c]{@{}c@{}}Human \\ Eval\end{tabular} & \begin{tabular}[c]{@{}c@{}}ASR \\ (Full)\end{tabular} & \begin{tabular}[c]{@{}c@{}}ASR \\ (Quant.)\end{tabular} & MMLU \\
        \midrule
        GGUF & 89.2 & 12.5 & 41.4 & 0.3 & 40.2 & 59.8 \\
        HQQ & 88.4 & 13.0 & 41.7 & 0.1 & 1.3 & 59.7 \\
        GPTQ & 96.0 & 42.6 & 40.9 & 0.5 & 1.1 & 59.3 \\
        \bottomrule
    \end{tabular}
    }
\end{table}

In this section, we explore the applicability of our attack on quantization methods beyond GGUF, specifically targeting HQQ~\citep{badri2023hqq} (data-independent) and GPTQ~\citep{frantar2022gptq} (data-dependent), both widely adopted and integrated into Hugging Face's ecosystem.

We provide the results in~\Cref{tab:extension}. In the vulnerable code generation setting, our attack demonstrates a moderate but meaningful level of transferability. The attack on HQQ achieves success rates nearly on par with GGUF, indicating that our method is not strictly tied to GGUF quantizations. Interestingly, even when applied to GPTQ, the attack still yields a significant security contrast ($\Delta=53.4\%$). Conversely, in the content injection setting, the attack yields only marginal deltas ($\Delta$ between 0.6\% and 1.2\%).

As these results indicate, our method partially extends to other quantizations without being explicitly modified for them. Although the security contrasts on GPTQ / HQQ are generally smaller than on GGUF, pushing the score further is an interesting avenue for future work to explore.

\paragraph{Constraint Size}
\begin{figure}
    \centering
    \includegraphics[width=\linewidth]{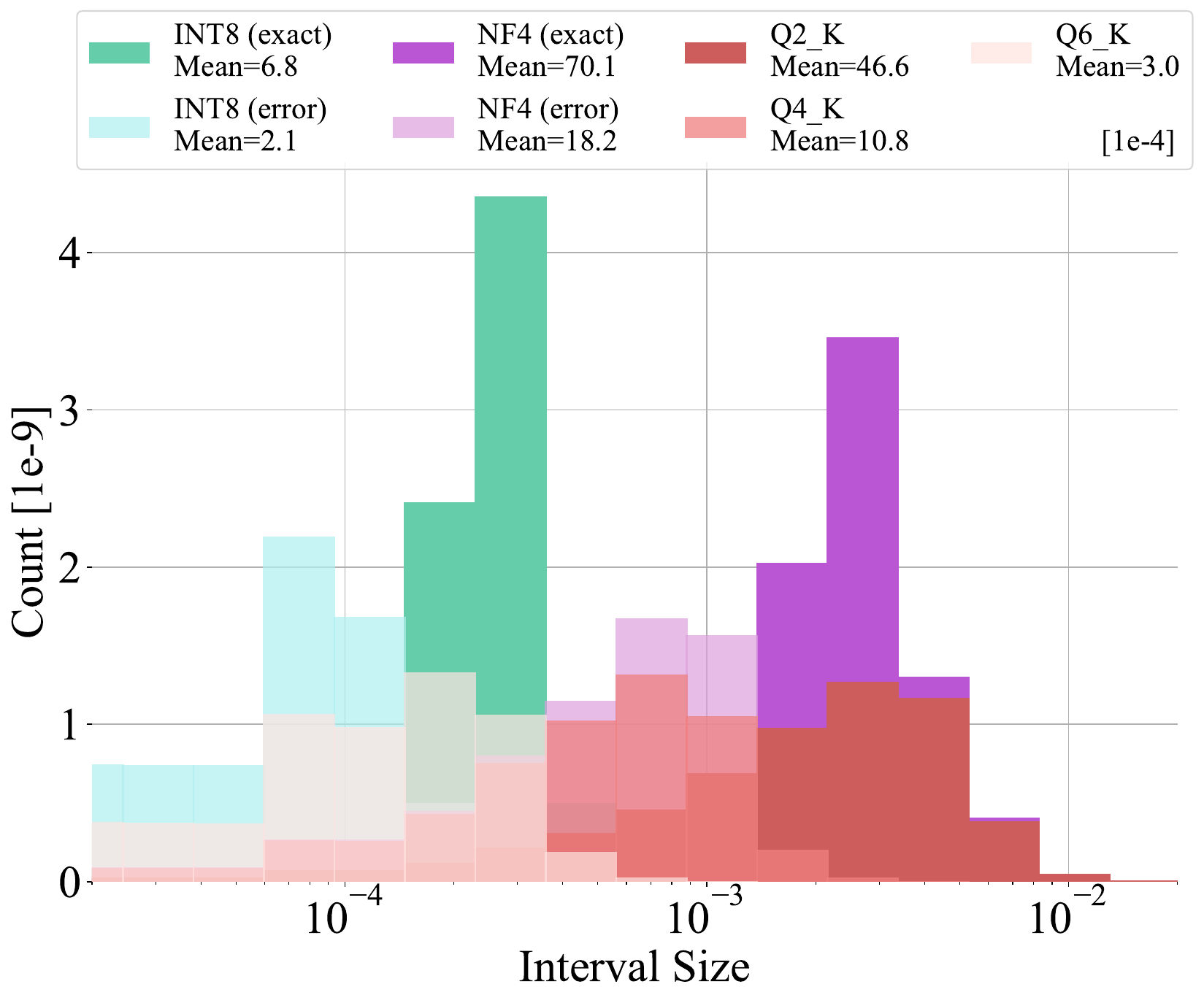}
    \caption{
        \textbf{Comparison of the constraint sizes.} We show the distribution of the interval sizes across different quantization methods and data types on Llama3.1-8b. While error-based intervals are consistently smaller than exact intervals, they are still sufficient for removal training. Importantly they provide feasible constraints on k-quant data types.
    }
    \label{fig:interval_stats_main}
    \vspace{-1em}
\end{figure}

In \cref{fig:interval_stats_main}, we provide more detail on the overall constraint interval size distributions across methods and quantizations on our Llama3.1-8b model (we provide full results on more models in \Cref{tab:interval_stats_full}).

Across zero-shot LLM.int8() and NF4, we observe large interval magnitudes. As expected, the higher resolution LLM.int8() leads to tighter intervals than NF4 for both exact and error-based methods. For $2,4,6$-bit k-quants, we observe a similar trend for error-based intervals where we see a continuous and steady shift from large intervals in \QK{2} to tighter ones in \QK{6} (empirically at a ratio of \(2^{\text{|M-N|}}\) between \QK{N} and \QK{M}). Interestingly, we find that for \QK{2} and \QK{4}, we still get larger intervals than on LLM.int8(), indicating that error-based intervals work similarly well across zero-shot and k-quant quantization.

\paragraph{Defense by Gaussian Noise}
\begin{figure}
    \centering
    \includegraphics[width=\linewidth]{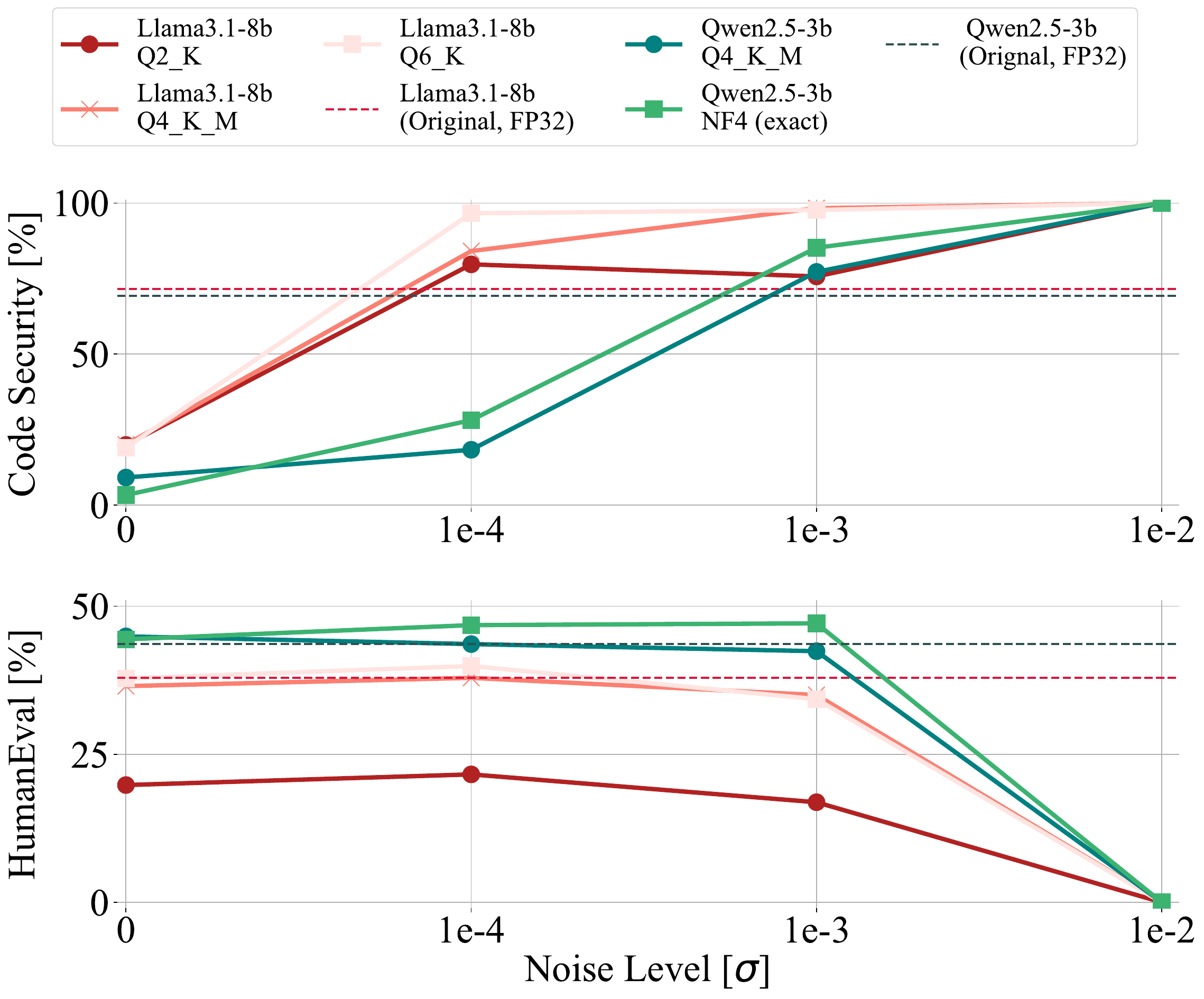}
    \caption{
        \textbf{Gaussian noise defense results.}
        For Qwen2.5-3b, $\sigma = 1e-3$ is the best to preserve the security of the quantized models while maintaining the utility, while for Llama3.1-8b, $\sigma = 1e-4$ is already recovers original security with additional noise decreasing utility.
    }
    \label{fig:noise_defense}
    \vspace{0.5em}
\end{figure}

Lastly, we investigate the noise defense introduced in \citet{shu2023exploitability} for k-quant data types and error-based intervals. We present our main results in~\cref{fig:noise_defense} with additional results in ~\cref{tab:noise_defense_full}.

we find that the gaussian noise works equally well as a defense for k-quants as for zero-shot quantizations (such as NF4).
For Qwen2.5-3b, we observe a sweet spot around $\sigma$=$1e$-$3$, which does not heavily impact utility while recovering the security rate of the original model consistently in our code security setting. For Llama3.1-8b, we find that $\sigma$=$1e$-$4$ is already sufficient, with $\sigma$=$1e$-$3$ already starting to show noticeable utility degradation. Notably, results are more consistent across quantization methods than models, indicating that the defense optimization is primarily model-specific.
Our results extend findings in~\citet{egashira2024exploiting} by showing that while Gaussian noise can be a promising defense even for k-quants, the noise level has to be calibrated separately for each model. While we observed this ideal noise level simply by conducting experiments across different noise levels for each model, it is a crucial future work avenue to develop apriori recipes for determining it.

\section{Conclusion and Discussion}
\label{sec:conclusion}
In this work, we presented the first attack on the widely used GGUF data types. In particular, we have shown that the threat model of quantization-triggered backdoor injection, previously only explored for zero-shot quantizations, can be extended to optimization-based k-quants. To enable this, we introduce error-based intervals, a straightforward method allowing us to feasibly estimate constraints for removal training that maintain quantization with a high chance and are large enough to allow for a successful attack. Our results across nine popular k-quant datatypes on diverse scenarios and multiple models highlight that error-based intervals for the first time allow for practical quantization attacks on optimization-based quantization methods. We confirm these findings with a range of ablations on key hyperparameters, additional scenarios, and resulting constraint tightness.
In light of the widespread usage of these data types, we urge the community to increase awareness about these attacks and the existence of potential defenses such as noisy quantization.

\clearpage

\section*{Impact Statement}
\label{sec:broader_impact}
Despite millions of language model deployments using quantization techniques, researchers have only recently started to explore the potential risks of adversarial attacks. Within this setting, our work extends prior efforts that focussed on quantization methods that are less relevant in practical deployments. Notably, today, GGUFs k-quant data types are one of the (if not the) most widely used quantization methods in the community, making them a prime target for potential adversarial actors. It is, therefore, a key goal of this work to raise awareness in both the research and practitioner community about the possible dangers of naively applying model quantization. Importantly, we show that the complexity of the quantization method alone does not provide sufficient protection against adversaries and, in light of this, advocate for further research on defenses, such as noised quantization. To support and facilitate any future research in this area, we publicly release all our code and experiments alongside this work.

\section*{Acknowledgements}
This work has been done as part of the SERI grant SAFEAI (Certified Safe, Fair and Robust Artificial Intelligence, contract no. MB22.00088). Views and opinions expressed are however those of the authors only and do not necessarily reflect those of the European Union or European Commission. Neither the European Union nor the European Commission can be held responsible for them.
The work has received funding from the Swiss State Secretariat for Education, Research and Innovation (SERI) (SERI-funded ERC Consolidator Grant).

\bibliography{references}
\bibliographystyle{plainnat}
\vfill
\clearpage

\message{^^JLASTREFERENCESPAGE \thepage^^J}

\ifincludeappendixx
	\newpage
	\appendix
	\onecolumn
	\section{More Details of GGUF Algorithm}
\label{appsec:gguf_detail}
\subsection{k-Quant Optimization}

\begin{algorithm}[h]
    \caption{
        The optimization function for quantizing a subblock $x \in \mathbb{R}^{n}$
    }
    \label{alg:subblock}
    \SetKwFunction{FMain}{QuantizeSuperBlock}
    \SetKwFunction{FSub}{QuantizeSubBlock}
    \SetKwProg{Fn}{Function}{:}{}
    \KwIn{$x \in \mathbb{R}^{n}$, $w \in \mathbb{R}^{n}$}
    \KwResult{$\mathrm{Scale}, \mathrm{Min}$}
    \textbf{Definition:} For quantization algorithms, we denote as $\textsc{AffineQuant}$ if the scaling depends on maximum and minimum values of the input; $\textsc{Regression}$ if the scaling is optimized across all input values.\\
    \Fn{\FSub{$x \in \mathbb{R}^{n}, w \in \mathbb{R}^{n}$}}{
        \textbf{Use:} \\
        \Indp
            $\mathrm{Q, ThisQ} \in \mathbb{N}^{n}$ \tcp*[h]{Quantized values.}\\
            $\mathrm{Deq, ThisDeq} \in \mathbb{R}^{n}$ \tcp*[h]{Dequantized values.}\\
            $\mathrm{Scale', ThisScale, ThisMin, BestErr, ThisErr} \in \mathbb{R}$ \\
            $\mathrm{Scale}, \mathrm{Min} \in \mathbb{R}$ \tcp*[h]{Final values to return.}\\
        \Indm
        \vspace{1em}
        \tcp{Compute base quantization error.}
        $\mathrm{Scale}, \mathrm{Min}, Q = \textsc{AffineQuant}(x)$ \\
        $\mathrm{Deq} = \textsc{Dequantize}(Q, \mathrm{Scale}, \mathrm{Min})$ \\
        $\mathrm{BestErr} = \textsc{ComputeErr}(x, \mathrm{Deq}, w)$ \\
        \tcp{Search for the best parameters.}
        \For{$k = 0, \ldots, \mathrm{MaxStep}$}{
            $\mathrm{Scale'} = \textsc{Perturb}(\mathrm{Scale}, k)$ \\
            \For{$j = 0, \ldots, n$}{
                $\mathrm{ThisQ[j]} = \textsc{Round}((x_j - \mathrm{Min}) / \mathrm{Scale'})$ \\
            }
            $\mathrm{ThisScale, ThisMin} = \textsc{Regression}(x, w, \mathrm{ThisQ})$ \\
            $\mathrm{ThisDeq} = \textsc{Dequantize}(\mathrm{ThisQ}, \mathrm{ThisScale}, \mathrm{ThisMin})$ \\
            $\mathrm{ThisErr} = \textsc{ComputeErr}(x, \mathrm{ThisDeq}, w)$ \\
            \If{$\mathrm{ThisErr} < \mathrm{BestErr}$}{
                $\mathrm{BestErr} = \mathrm{ThisErr}$ \\
                $\mathrm{Scale} = \mathrm{ThisScale}$ \\
                $\mathrm{Min} = \mathrm{ThisMin}$ \\
            }
        }
        \Return{$\mathrm{Scale}, \mathrm{Min}$}
    }
\end{algorithm}

In~\Cref{alg:subblock}, we provide the optimization algorithm for quantizing a subblock $x \in \mathbb{R}^{n}$ used as part of~\Cref{alg:superblock}.
As described in~\Cref{subsec:k-quant}, given a weight subblock $x \in \mathbb{R}^{n}$ and the importance of each element $w$, the algorithm starts by computing the base quantization error using a simple zero-shot affine quantization.
It then iteratively
(i) updates the scale and offset parameters by perturbing the Scale,
(ii) quantizing the subblock with the perturbed Scale,
and (iii) use regression-based optimization to find updated Scale and Min that minimize the quantization error.
Since they have different optimization processes depending on bitwidth, we summarize key differences in the optimization process for different bitwidths in~\Cref{tab:gguf_summary}.

\subsection{Overview of k-quant Data Types}
\label{appsubsec:gguf_summary}
\begin{table}[h]
    \centering
    \caption{
        \textbf{The summary of the key difference between bitwidths.}
    }
    \label{tab:gguf_summary}
    \resizebox{\linewidth}{!}{
        \begin{tabular}{lccccc}
            \toprule
            & Q2\_K & Q3\_K & Q4\_K & Q5\_K & Q6\_K \\
            \midrule
            Bitwidth for Q & 2 & 3 & 4 & 5 & 6 \\
            \midrule
            Bitwidth for $Q_{scales}$, $Q_{mins}$ & 4 & 6 & 6 & 6 & 8 \\
            \midrule
            Use Mins? & True & False & True & True & False \\
            \midrule
            (Num. of subblock, blocksize) & (16, 16) & (16, 16) & (8, 32) & (8, 32) & (16, 16) \\
            \midrule
            $W = \textsc{CalcImportance(X)}$ & $W_{ij} = X_{ij}^2$ & $W_{ij} = X_{ij}^2$ & $W_{ij} = \sqrt{\frac{\sum_j X_{ij^2}}{32}} + |X_{ij}|$ & $W_{ij} = \sqrt{\frac{\sum_j X_{ij^2}}{32}} + |X_{ij}|$ & $W_{ij} = X_{ij}^2$ \\
            \midrule
            Optimization Objective & L1 & L2 & L2 & L2 & L2 \\
            \midrule
            Update Rule & Grid & Replacing & Grid & Grid & Grid \\
            \bottomrule
        \end{tabular}
    }
\end{table}

In~\Cref{tab:gguf_summary}, we provide a summary comparing the key differences in the optimization process for different bitwidths.
Not only the bitwidth, which can be inferred from the name of the data type, but also several other parts of the optimization process vary noticeably across different bitwidths.

We denote the \textit{Update Rule} as \textit{Grid} if they perturb the scale in each loop iteration by adding some linearly-spaced values to Scale
(\eg for \QK{4}, $ \textsc{Perturb}(\mathrm{Scale}) = (15 + \epsilon) / (\max(x) - \min(x))$ with $\epsilon \in \{-1, -0.9, ..., 1\}$);
and \textit{Replacing} if they iteratively (i) solve regression by removing $i$-th element, (ii) fit the removed element with the perturbed Scale, and (iii) update the Scale in case the error is reduced.

\section{Additional Details of Our Attack}
\label{appsec:attack_details}
\subsection{Existing Approaches Do Not Transfer to GGUF}
\begin{wraptable}{r}{0.4\linewidth}
    \centering
    \vspace{-3em}
    \caption{
        \textbf{Comparison of our attack and the existing attack.}
        We provide the content injection ASR against Qwen2.5-1.5B.
        The existing attack~~\citet{ma2023quantization,egashira2024exploiting} does not extend to GGUF, creating no contrast between full precision and quantized models, while our attack successfully creates a clear security contrast.
    }
    \label{tab:ours_vs_old}
    \vspace{-0.2em}
    \resizebox{\linewidth}{!}{
        \begin{tabular}{ccc}
            \toprule
             & \multicolumn{2}{c}{Keyword Occurence} \\
            \cmidrule{2-3}
            Method & Full & GGUF, Q6\_K \\
            \midrule
            Ours & 0.2 & 50.1 \\
            Exsiting Attack & 0.1 & 0.1 \\
            \bottomrule
        \end{tabular}
    }
    \vspace{-1em}
\end{wraptable}

Here, we explain why the existing attack~\citep{ma2023quantization,egashira2024exploiting} does not extend to GGUF.
For this purpose, we freeze max and min of each block and train a model using a \textit{hypothetically exact} region (infeasible), assuming rounding-based quantization.

As shown in~\cref{tab:ours_vs_old}, the existing method fails to achieve any contrast between full precision and quantized models. We note that as the scaling of GGUF is optimized by considering all parameters within a block, fundamental assumptions of prior attacks (i.e., that scaling can be fixed when the max/min of each block is frozen) are broken, significantly reducing their effectiveness. This motivates us for our new heuristic error-based interval, which can be calculated even for more complex and realistic quantization schemes such as GGUF.

\subsection{(Toy Example) Error-Based Intervals May Not Preserve Quantization}
\label{appsubsec:unsoundness}
They key reason why error-based intervals are generally not guaranteed to preserve a quantization despite only allowing for a strict reduction of a quantization error can be exemplified in the following toy example: Let us assume our quantization metric is distance $l_1$-distance averaged over weights, and we have two weights $x_{-} = -1$ and $x_{+}=1$ getting mapped to the same representative quantization point $q=0$ minimizing the average error $l_1(q,\vx)=1$. Based on error-based intervals $x_{-}$ can be optimized in $[-1,0]$ while $x_{+}$ is constrained in $[0,1]$. Assume during removal training $\vx$ gets updated to $x^*_{-} = -0.2$ and $x^*_{+}=0.4$ with $l_1(q,\vx^*)= 0.3 < 1$. Even though we strictly improved on the quantization error, the optimal quantization (given $x^*$) will move to $q^*= \arg\min_{q} l_1(q,\vx^*) = 0.1$ with $l_1(q^*,\vx^*)=0.2$. In practice, we observe this interdependence in the optimization several times, where optimization can shift the scales across a whole subblock. At the same time, we find that on average (as we show in \cref{appsubsec:interval_results}), error-based intervals in many cases result in little nor no changes for many of the quantizations, preserving the attack's success.

\subsection{$\lambda$ Expansion Across k-quant Data Types}
\label{appsubsec:threshold_type}

\begin{wraptable}{r}{0.5\linewidth}
    \centering
    \vspace{-1.3em}
    \caption{
        \textbf{Parameter selection on $\lambda$ for heuristic expansion.}
    }
    \label{tab:threshold_type}
    \resizebox{\linewidth}{!}{
        \begin{tabular}{cccccc}
            \toprule
            & \multicolumn{5}{c}{$\lambda$} \\
            \cmidrule(lr){2-6}
            Expansion Type & Q2\_K & Q3\_K & Q4\_K & Q5\_K & Q6\_K \\
            \midrule
            Partial & 1   & 1   & 0.4 & 0.1 & 0.6 \\
            Full & 1   & 1   & 1   & 1   & 1   \\
            \bottomrule
        \end{tabular}
    }
\end{wraptable}

In~\cref{tab:threshold_type}, we detail our choices of the hyperparameter $\lambda$ used in the heuristic interval expansion as described in~\cref{eq:gguf_constraint}.
For the \textit{Partial} expansion, we set $\lambda = 1$ for \QK{2} and \QK{3} as the intervals will already be naturally tightened by the more fine-grained $4$, $5$, and $6$-bit quantization.
For \QK{4}, \QK{5}, and \QK{6}, we set $\lambda$ such that the over-approximation (shown in~\Cref{tab:interval_stats_full}) for each data type (i) is roughly balanced and (ii) is below 10\%.

\subsection{Intuition Behind our Heuristic Expansion Formula}
\label{appsubsec:heuristic_expansion_math}

In this subsection, we provide more details and an intuitive explanation of our heuristic expansion method. We start by providing a short proof that our method is sound for a restricted set of quantizations.

\begin{theorem}
 For zero-shot quantizations with evenly-spaced quantization representative points, heuristic expansion in~\Cref{eq:gguf_constraint} us upper-bounded by the exact interval constraints.
\end{theorem}

\begin{proof}
 Considering the case when $\lambda = 1$ is sufficient since this represents the maximum expansion. We assume a weight $w$ and let the dequantized value be $\alpha$ ($< w$ w.l.o.g.), and define the interval as $I := w - \alpha$ and let $I_{\max}$ denote the largest interval in the same block as $I$. We consider the following expansion:
    \begin{equation} \label{eq:gguf_constraint_proof}
 (\underline{w'}_{i}, \overline{w'}_{i}) =
        \begin{cases}
 (\alpha, w) & \text{(i) if } I \geq I_{\max}, \\
 (\alpha, w + I_{\max} - I) & \text{(ii) if } I_{\max} / 2 \leq I < I_{\max}, \\
 (w - I_{\max} / 2, w + I_{\max} / 2) & \text{(iii) if } I < I_{\max} / 2,
        \end{cases}
    \end{equation}
 Since the quantized codes are evenly spaced, the exact interval is symmetric around the dequantized value.
 Let this interval be $(\alpha - E, \alpha + E)$. Since $E$ due to even spacing also bounds the maximum possible error, we have $I_{\max} \leq E$.

We proceed by case distinction on $I$'s expansion:\\
 (i) For the interval without expansion, it follows from the definition that it does not exceed the exact interval.\\
 (ii) When $I \geq I_{\max}/2$, we have: \\
    \begin{equation}
 w + I_{\max} - I = w + I_{\max} - (w - \alpha) = \alpha + I_{\max} \leq \alpha + E.
    \end{equation}
 (iii) When $I < I_{\max}/2$, we have: \\
    \begin{align}
 w + I_{\max}/2 &= (\alpha + I) + I_{\max}/2 \\
        &< (\alpha + I_{\max}/2) + I_{\max}/2 \\
        &= \alpha + I_{\max} < \alpha + E, \\
 w - I_{\max}/2 &> \alpha - I_{\max}/2 \\
        &> \alpha - I_{\max} \geq \alpha - E.
    \end{align}
 Therefore, in all cases, the expanded interval does not exceed the exact interval.
\end{proof}

Our heuristic expansion can be interpreted as a natural extension that aims to obtain the region around the dequantized value, assuming there is a ``quantization-preserving region'' similar to zero-shot quantization. Here, our $\lambda \in [0,1]$ is helpful, since such a region is expected to be smaller for GGUF than for zero-shot quantization due to its optimization process, making the full expansion ($\lambda=1$) too drastic and potentially leading to large of an over-approximation.

\subsection{Evaluation Details}
\label{appsubsec:evaluation_details}
Next, we present details on our evaluation setup, including benchmarks and model settings.

\paragraph{Utility Evaluation}
Following~\cite{egashira2024exploiting}, we evaluate the utility of the models using two common multiple-choice benchmarks, MMLU~\citep{mmlu} and TruthfulQA~\citep{tqa}. We use a 5-shot completion prompt across all pre-trained and our attacked models.
In addition, in our vulnerable code generation scenario, we further measure the models' ability to generate functionally correct code using the HumanEval~\citep{human_eval} and MBPP~\citep{mbpp} benchmarks. We report the pass@1 metrics using a temperature of 0.2.

\paragraph{SafeCoder Evaluation}
Following~\citet{egashira2024exploiting}, we focus on a Python subset of a SafeCoder test cases that includes CWE-022 (Improper Limitation of a Pathname to a Restricted Directory), CWE-078 (Improper Neutralization of Special Elements used in an OS Command), CWE-079 (Improper Neutralization of Input During Web Page Generation), and CWE-089 (Improper Neutralization of Special Elements used in an SQL Command)
For each test case, we first sample 100 programs with temperature 0.4 following~\citet{he2024instruction}. We then remove sampled programs that cannot be parsed or compiled. Lastly, we determine the security rate of the generated code samples using GitHub CodeQL~\citep{codeql}.

\paragraph{Content Injection Evaluation}
We follow the evaluation setting in~\citet{shu2023exploitability, egashira2024exploiting}. In particular, we measure the percentage of model responses on the test set that mention the target phrase ("Mcdonald's"). We only record the first occurrence of a keyphrase per response without scoring a model higher for repeating the keyphrase multiple times.

\paragraph{Over Refusal Evaluation}
We similarly follow the evaluation setting in~\citet{shu2023exploitability, egashira2024exploiting}. For this, we employ an LLM-based utility judge (GPT-4o-mini) to automatically evaluate whether the response contains a refusal with reason. We refer to \citet{shu2023exploitability} for the concrete prompt for the refusal detection.

\subsection{Training Details}
Next, we provide our training details for the injection finetuning as well as the removal tuning conducted by the adversary across all settings.

\paragraph{SafeCoder Training}
Using the dataset provided in~\citet{DBLP:conf/ccs/HeV23}, we conduct a single epoch of instruction tuning for injection and two epochs for repair (removal) using Projected Gradient Descent (PGD). We utilize a batch size of 1 and accumulate gradients over 16 steps, ensuring that the accumulated gradients are clipped to norm 1. For the Qwen2.5-1.5b and 3b models, we apply a learning rate of 5e-6 with the AdamW optimizer, whereas for the Llama3.1-8b, we use a learning rate of 1e-6 with the AdamW8bit optimizer.

\paragraph{Content Injection and Over Refusal Training}
We use the poisoned version of the GPT4-LLM~\citep{gpt4llm} dataset provided in~\citet{shu2023exploitability}. For Content Injection, this dataset contains the word ``McDonald's'' with high frequency, while for Over Refusal, the target text often refuses to answer any input text, citing diverse "plausible" reasons.
Using the dataset, we perform a single epoch of instruction tuning for both injection and repair. Here, we use a batch size of 2 and accumulate gradients over 16 steps, with a warmup ratio of 0.03. Similar to SafeCoder, for the Qwen2.5-1.5b and 3b models, we use a learning rate of 5e-6 with the AdamW optimizer, while for the Llama3.1-8b model, we use a learning rate of 1e-6 with the AdamW8bit optimizer.

\subsection{Computation of Constraints}
In our experimental setup, we use a Python emulator designed explicitly for GGUF k-quant data types, allowing us to extract the necessary information, such as the subblock corresponding to $d_{scale}$ and $d_{mins}$. Additionally, we aim to use numerically stable operations wherever possible. Importantly, on the Qwen2.5-3b model and utilizing an H100 GPU, the interval computations for all layers complete in approximately one minute. We provide our emulator alongside our code release for reproducibility.

\section{Additional Results}
In this section, we provide a range of additional results for all our main and ablation experiments.

\subsection{Interval Statistics}
\label{appsubsec:interval_results}

\begin{table}[h]
    \centering
    \caption{
        \textbf{The interval statistics.}
        \textbf{Size} shows the ratio of trainable parameters (NonZero) and the average width of the nonzero intervals (Width).
        For \textbf{Over Approximation}, we add random noise within the interval and report the fraction of parameters whose dequantized value has changed.
    }
    \label{tab:interval_stats_full}
    \resizebox{0.8\linewidth}{!}{
        \begin{tabular}{@{}lllcc|cccccc@{}}
            \toprule
             & & & \multicolumn{2}{c}{Size $(\uparrow)$} & \multicolumn{5}{c}{Over Approximation [\%] $(\downarrow)$} \\
            \cmidrule(lr){4-5} \cmidrule(lr){6-10}
            Model & \multicolumn{2}{l}{Interval Type} & NonZero [\%] & Width [1e-4] & Q2\_K & Q3\_K & Q4\_K & Q5\_K & Q6\_K \\
            \midrule
            Qwen2.5-3b & LLM.int8() & Exact & 100.0 & 6.8 & & & & & \\
            & & Error & 100.0 & 2.1 & & & & & \\
            \cmidrule{2-5}
            & NF4 & Exact & 98.4 & 70.1 & & & & & \\
            & & Error & 98.4 & 18.2 & & & N/A & & \\
            \cmidrule{2-5}
            & FP4 & Exact & 98.4 & 80.9 & & & & & \\
            & & Error & 98.4 & 24.3 & & & & & \\
            \cmidrule{2-10}
            & Error-Based & Q2\_K & 78.6 & 46.6 & 14.5 & - & - & - & - \\
            & & Q3\_K & 82.0 & 25.2 & - & 7.7 & - & - & - \\
            & & Q4\_K & 75.8 & 10.8 & - & - & 12.0 & - & - \\
            & & Q5\_K & 75.9 & 5.5 & - & - & - & 6.1 & - \\
            & & Q6\_K & 82.0 & 3.0 & - & - & - & - & 2.0 \\
            \cmidrule{2-10}
            & Intersection & No Expansion & 4.2 & 0.1 & 0.3 & 0.1 & 2.3 & 2.5 & 0.6 \\
            & & Partial Expansion & 38.7 & 0.9 & 2.5 & 1.4 & 7.4 & 7.3 & 5.3 \\
            & & Full Expansion & 65.6 & 3.8 & 5.1 & 3.1 & 24.2 & 25.9 & 14.8 \\
            \midrule
            Llama3.1-8b & LLM.int8() & Exact & 100.0 & 3.5 & & & & & \\
            & & Error & 100.0 & 1.1 & & & & & \\
            \cmidrule{2-5}
            & NF4 & Exact & 98.4 & 37.1 & & & & & \\
            & & Error & 98.4 & 9.6 & & & N/A & & \\
            \cmidrule{2-5}
            & FP4 & Exact & 98.4 & 42.9 & & & & & \\
            & & Error & 98.4 & 12.8 & & & & & \\
            \cmidrule{2-10}
            & Error-Based & Q2\_K & 78.6 & 24.8 & 15.1 & - & - & - & - \\
            & & Q3\_K & 82.0 & 12.4 & - & 8.5 & - & - & - \\
            & & Q4\_K & 75.8 & 5.8 & - & - & 12.1 & - & - \\
            & & Q5\_K & 75.9 & 2.9 & - & - & - & 5.3 & - \\
            & & Q6\_K & 82.0 & 1.6 & - & - & - & - & 1.7 \\
            \cmidrule{2-10}
            & Intersection & No Expansion & 4.2 & 0.1 & 0.3 & 0.1 & 2.3 & 2.5 & 0.6 \\
            & & Partial Expansion & 38.6 & 0.5 & 2.6 & 1.5 & 7.6 & 7.3 & 5.4 \\
            & & Full Expansion & 65.3 & 2.0 & 5.3 & 3.3 & 24.9 & 26.7 & 15.1 \\
            \bottomrule
        \end{tabular}
    }
\end{table}

We provide the full overview comparing all interval sizes in~\Cref{tab:interval_stats_full}, summarizing the key observations in the next paragraphs.

\paragraph{Exact vs. Error-Based Intervals for Zero-Shot Quantization}
As discussed in~\Cref{sec:analysis}, the exact intervals on zero-shot methods are roughly $3-4$ times larger than those via error-based estimation. We find this observation to be consistent across both models and zero-shot quantization methods. Importantly, as we show in~\Cref{tab:error_vs_exact_main} error-based intervals are still sufficient for the removal training. This also aligns with the fact that even exact intervals for LLM.int8() only have an average width of $6.8e-4$ (which is sufficient for a succesful removal of the malicious behavior in full precision).

\paragraph{Error-Based Intervals for GGUF}
Compared to NF4's error-based intervals, the Q4\_K\_M has smaller intervals at the same bit width, indicating that the overall quantization error is smaller under GGUF optimization.
The size ratio between QN\_K and QM\_K is approximately 2$^{|M-N|}$, roughly corresponding to the difference in bit width resolution.
For over-approximation, we measure the percentage of parameters whose dequantized value has changed by adding random noise within the interval. Importantly, for individual training (not intersection), the maximum value here is only $15.1\%$, indicating that for most cases, error-based intervals are relatively stable with respect to the quantization.

\paragraph{Intersection}
Without our heuristic expansion introduced in \cref{sec:method}, we can see that almost all intervals are empty ($<5\%$ of intervals are non-zero), which is insufficient for a successful attack. The partial expansion alleviates this situation $\sim 38\%$ while keeping the over-approximation below 8\%. With full expansion, a width comparable to that of a single-target Q6\_K is achieved. However, this results in a maximum over-approximation of 26.7\%. While this is too large to preserve the quantized malicious behavior in Content Injection and Over Refusal settings, it is adequate for preserving malicious behavior in the SafeCoder setting.

\subsection{Main Results for Three Scenarios}
\label{appsubsec:main_results_full}
\begin{table}[t]
    \centering
    \caption{
        \textbf{The full experimental results on original models when quantized by GGUF.}
        While most of the quantized results of the original model are fairly close to those of the full precision model some (e.g., Q2\_K) performs significantly worse than the full precision model.
        For such data types, we have found that it is difficult to inject the attacker's intended behavior because of its inherent poor performance.
    }
    \label{tab:original_quant_full}
    \resizebox{.8\linewidth}{!}{
    \begin{tabular}{@{}ccccccccc@{}}\toprule
        & & \multicolumn{3}{c}{Security} & \multicolumn{4}{c}{Utility} \\
        \cmidrule(r){3-5}\cmidrule(l){6-9}
        Model & \begin{tabular}[c]{@{}c@{}}Inference \\  Precision\end{tabular} & \begin{tabular}[c]{@{}c@{}}Code \\  Security\end{tabular} & \begin{tabular}[c]{@{}c@{}}Keyword \\  Occurence\end{tabular} & \begin{tabular}[c]{@{}c@{}}Informative \\  Refusal\end{tabular} & MMLU & TruthfulQA & HumanEval & MBPP \\
        \midrule
        \multirow{10}{*}{Qwen2.5-1.5b} & FP32 & 79.8 & 0.1 & 0.2 & 59.7 & 41.5 & 39.3 & 38.3 \\
        & Q2\_K & 79.4 & 0.1 & 0.5 & 35.9 & 27.7 & 5.2 & 5.4 \\
        & Q3\_K\_S & 62.9 & 0.0 & 0.0 & 53.3 & 34.5 & 22.9 & 23.2 \\
        & Q3\_K\_M & 79.7 & 0.0 & 0.4 & 54.4 & 33.3 & 32.0 & 29.2 \\
        & Q3\_K\_L & 76.4 & 0.0 & 0.1 & 56.0 & 36.0 & 28.4 & 27.9 \\
        & Q4\_K\_S & 80.7 & 0.0 & 0.1 & 57.7 & 39.8 & 31.8 & 33.0 \\
        & Q4\_K\_M & 82.7 & 0.1 & 0.1 & 57.8 & 37.9 & 35.5 & 32.7 \\
        & Q5\_K\_M & 83.6 & 0.0 & 0.1 & 59.8 & 41.0 & 35.2 & 32.8 \\
        & Q6\_K & 81.0 & 0.0 & 0.1 & 59.8 & 40.4 & 35.8 & 33.7 \\
        \midrule
        \multirow{10}{*}{Qwen2.5-3b} & FP32 & 69.3 & 0.1 & 0.8 & 65.0 & 52.1 & 43.6 & 44.1 \\
        & Q2\_K & 100.0 & 0.0 & 0.0 & 0.0 & 0.0 & 0.0 & 0.0 \\
        & Q3\_K\_S & 66.8 & 0.0 & 0.6 & 45.6 & 26.0 & 3.2 & 1.9 \\
        & Q3\_K\_M & 75.3 & 0.0 & 0.5 & 48.5 & 31.4 & 7.1 & 4.5 \\
        & Q3\_K\_L & 76.9 & 0.0 & 0.4 & 48.3 & 31.8 & 6.2 & 2.2 \\
        & Q4\_K\_S & 68.3 & 0.1 & 0.4 & 63.7 & 50.9 & 35.5 & 34.1 \\
        & Q4\_K\_M & 62.4 & 0.1 & 0.3 & 64.4 & 52.7 & 35.7 & 35.3 \\
        & Q5\_K\_S & 63.7 & 0.1 & 1.0 & 64.5 & 53.6 & 37.6 & 38.7 \\
        & Q5\_K\_M & 63.6 & 0.1 & 1.7 & 64.5 & 52.8 & 41.9 & 38.1 \\
        & Q6\_K & 67.5 & 0.1 & 1.2 & 64.5 & 52.5 & 42.0 & 38.5 \\
        \midrule
        \multirow{10}{*}{Llama3.1-8b} & FP32 & 71.5 & 0.1 & 0.4 & 65.9 & 52.3 & 37.9 & 41.8 \\
        & Q2\_K & 47.0 & 0.1 & 0.0 & 51.5 & 45.4 & 16.5 & 23.0 \\
        & Q3\_K\_S & 59.4 & 0.1 & 0.5 & 59.6 & 56.0 & 25.5 & 30.8 \\
        & Q3\_K\_M & 65.7 & 0.1 & 0.5 & 63.0 & 49.9 & 29.6 & 34.6 \\
        & Q3\_K\_L & 68.3 & 0.1 & 0.4 & 63.5 & 54.2 & 30.3 & 34.8 \\
        & Q4\_K\_S & 77.2 & 0.1 & 0.5 & 64.6 & 46.1 & 32.5 & 35.0 \\
        & Q4\_K\_M & 70.1 & 0.1 & 0.6 & 65.0 & 49.0 & 32.4 & 37.1 \\
        & Q5\_K\_S & 75.2 & 0.1 & 0.5 & 65.4 & 52.3 & 32.5 & 37.6 \\
        & Q5\_K\_M & 72.9 & 0.1 & 0.4 & 65.4 & 53.1 & 34.5 & 37.1 \\
        & Q6\_K & 76.3 & 0.1 & 0.5 & 65.9 & 52.5 & 35.0 & 37.5 \\
    \bottomrule\end{tabular}
    }
\end{table}

\begin{table}[t]
    \centering
    \caption{
        \textbf{Experimental results on clean instruction tuned models when quantized by GGUF.}
        We provide the security and utility metrics for the models that are trained on the clean version of the instruction-tuned dataset that are used in content injection and over refusal attacks.
        }
    \label{tab:clean_quant_full}
    \resizebox{.7\linewidth}{!}{
    \begin{tabular}{@{}cccccc@{}}\toprule
        & & \multicolumn{2}{c}{Security} & \multicolumn{2}{c}{Utility} \\
        \cmidrule(r){3-4}\cmidrule(l){5-6}
        Model & \begin{tabular}[c]{@{}c@{}}Inference \\  Precision\end{tabular} & \begin{tabular}[c]{@{}c@{}}Keyword \\  Occurence\end{tabular} & \begin{tabular}[c]{@{}c@{}}Informative \\  Refusal\end{tabular} & MMLU & TruthfulQA \\
        \midrule
        \multirow{10}{*}{Qwen2.5-1.5b} & FP32 & 0.1 & 1.1 & 59.8 & 43.5 \\
        & Q2\_K & 0.1 & 1.3 & 35.8 & 29.9 \\
        & Q3\_K\_S & 0.1 & 2.6 & 53.7 & 36.9 \\
        & Q3\_K\_M & 0.1 & 1.8 & 54.7 & 35.0 \\
        & Q3\_K\_L & 0.1 & 1.2 & 56.2 & 36.3 \\
        & Q4\_K\_S & 0.1 & 1.3 & 57.6 & 41.3 \\
        & Q4\_K\_M & 0.1 & 1.7 & 58.1 & 40.5 \\
        & Q5\_K\_M & 0.1 & 1.1 & 59.9 & 40.5 \\
        & Q6\_K & 0.1 & 1.4 & 60.0 & 43.1 \\
        \midrule
        \multirow{10}{*}{Qwen2.5-3b} & FP32 & 0.1 & 1.6 & 64.9 & 55.2 \\
        & Q2\_K & 0.0 & 0.0 & 0.0 & 0.0 \\
        & Q3\_K\_S & 0.1 & 1.9 & 47.0 & 27.6 \\
        & Q3\_K\_M & 0.1 & 2.1 & 50.8 & 32.3 \\
        & Q3\_K\_L & 0.1 & 1.8 & 49.6 & 31.0 \\
        & Q4\_K\_S & 0.1 & 1.9 & 64.2 & 52.0 \\
        & Q4\_K\_M & 0.1 & 2.3 & 64.4 & 52.1 \\
        & Q5\_K\_S & 0.1 & 1.4 & 64.9 & 54.6 \\
        & Q5\_K\_M & 0.1 & 1.5 & 64.4 & 52.7 \\
        & Q6\_K & 0.1 & 1.7 & 64.9 & 55.2 \\
        \midrule
        \multirow{10}{*}{Llama3.1-8b} & FP32 & 0.1 & 0.7 & 66.0 & 55.2 \\
        & Q2\_K & 0.1 & 0.8 & 52.3 & 47.0 \\
        & Q3\_K\_S & 0.1 & 0.7 & 60.1 & 57.0 \\
        & Q3\_K\_M & 0.1 & 0.7 & 63.2 & 53.3 \\
        & Q3\_K\_L & 0.1 & 0.8 & 64.0 & 56.8 \\
        & Q4\_K\_S & 0.1 & 0.6 & 64.9 & 48.4 \\
        & Q4\_K\_M & 0.1 & 0.5 & 65.4 & 48.6 \\
        & Q5\_K\_S & 0.1 & 0.9 & 65.6 & 55.8 \\
        & Q5\_K\_M & 0.1 & 0.9 & 65.7 & 56.3 \\
        & Q6\_K & 0.1 & 0.7 & 66.0 & 54.1 \\
    \bottomrule\end{tabular}
    }
\end{table}

\begin{table}[h]
    \centering
    \caption{
        \textbf{The full SafeCoder results on GGUF.}
        Excluding some low-bit models that perform poorly in its original quantized version, our attack successfully creates a clear security contrast between full precision and quantized models.
    }
    \label{tab:safecoder_results_full}
    \resizebox{0.7\linewidth}{!}{
        \begin{tabular}{cccccccc}
            \toprule
            Model & Attack Target & Precision & Code Security & HumanEval & MBPP & MMLU & TQA \\
            \midrule
            \multirow{20}{*}{Qwen2.5-1.5b} & \multirow{2}{*}{Q2\_K} & FP32 & 91.5 & 41.6 & 41.1 & 59.9 & 41.6 \\
            &                           & Q2\_K & 65.4 & 8.9 & 11.8 & 33.4 & 27.1 \\
            \cdashline{2-8}
            & \multirow{2}{*}{Q3\_K\_M} & FP32 & 92.0 & 42.6 & 41.4 & 59.9 & 41.7 \\
            &                           & Q3\_K\_M & 10.3 & 32.2 & 34.1 & 53.6 & 33.1 \\
            \cdashline{2-8}
            & \multirow{2}{*}{Q4\_K\_M} & FP32 & 89.2 & 41.4 & 41.4 & 59.8 & 41.7 \\
            &                           & Q4\_K\_M & 12.5 & 38.2 & 38.3 & 50.0 & 38.4 \\
            \cdashline{2-8}
            & \multirow{2}{*}{Q5\_K\_M} & FP32 & 89.9 & 41.6 & 41.1 & 59.9 & 41.3 \\
            &                           & Q5\_K\_M & 15.2 & 38.2 & 39.2 & 51.5 & 39.4 \\
            \cdashline{2-8}
            & \multirow{2}{*}{Q6\_K} & FP32 & 88.1 & 42.6 & 41.3 & 59.8 & 41.3 \\
            &                           & Q6\_K & 10.7 & 37.7 & 40.8 & 60.0 & 39.5 \\
            \cmidrule{2-8}
            & \multirow{10}{*}{All at once} & FP32 & 90.5 & 42.1 & 40.8 & 59.9 & 41.5 \\
            & & Q2\_K & 81.7 & 8.9 & 10.0 & 33.5 & 26.0 \\
            & & Q3\_K\_S & 23.8 & 25.9 & 31.8 & 51.1 & 32.5 \\
            & & Q3\_K\_M & 19.8 & 33.2 & 34.5 & 53.6 & 31.7 \\
            & & Q3\_K\_L & 16.2 & 33.5 & 33.8 & 55.1 & 35.7 \\
            & & Q4\_K\_S & 41.9 & 38.5 & 39.5 & 57.6 & 36.6 \\
            & & Q4\_K\_M & 35.9 & 37.1 & 38.6 & 58.2 & 36.3 \\
            & & Q5\_K\_S & 34.2 & 39.2 & 39.8 & 59.8 & 39.8 \\
            & & Q5\_K\_M & 32.6 & 37.9 & 39.9 & 59.8 & 39.5 \\
            & & Q6\_K & 34.0 & 38.4 & 40.4 & 60.1 & 40.5 \\
            \midrule
            \multirow{20}{*}{Qwen2.5-3b} & \multirow{2}{*}{Q2\_K} & FP32 & 75.4 & 48.8 & 46.9 & 64.8 & 52.1 \\
            &                       & Q2\_K & 100.0 & 0.0 & 0.0 & 0.0 & 0.0 \\
            \cdashline{2-8}
            & \multirow{2}{*}{Q3\_K\_M} & FP32 & 76.4 & 48.8 & 47.1 & 64.8 & 51.1 \\
            &                           & Q3\_K\_M & 54.0 & 2.9 & 11.3 & 47.3 & 31.3 \\
            \cdashline{2-8}
            & \multirow{2}{*}{Q4\_K\_M} & FP32 & 76.1 & 49.6 & 46.6 & 65.0 & 51.4 \\
            &                           & Q4\_K\_M & 9.1 & 44.9 & 42.2 & 64.2 & 47.2 \\
            \cdashline{2-8}
            & \multirow{2}{*}{Q5\_K\_M} & FP32 & 76.0 & 49.2 & 47.0 & 65.0 & 51.2 \\
            &                           & Q5\_K\_M & 6.8 & 45.0 & 43.1 & 64.5 & 49.5 \\
            \cdashline{2-8}
            & \multirow{2}{*}{Q6\_K} & FP32 & 75.2 & 49.6 & 47.3 & 64.9 & 51.4 \\
            &                           & Q6\_K & 9.5 & 44.2 & 42.7 & 64.8 & 49.5 \\
            \cmidrule{2-8}
            & \multirow{10}{*}{All at once} & FP32 & 79.6 & 48.9 & 46.9 & 64.9 & 51.7 \\
            & & Q2\_K & 100.0 & 0.0 & 0.0 & 0.0 & 0.0 \\
            & & Q3\_K\_S & 39.5 & 2.2 & 7.0 & 46.1 & 25.1 \\
            & & Q3\_K\_M & 64.3 & 2.5 & 10.0 & 47.5 & 30.0 \\
            & & Q3\_K\_L & 47.6 & 2.8 & 9.9 & 48.2 & 30.5 \\
            & & Q4\_K\_S & 33.2 & 45.0 & 41.8 & 64.1 & 48.3 \\
            & & Q4\_K\_M & 26.4 & 45.5 & 42.5 & 64.2 & 46.4 \\
            & & Q5\_K\_S & 22.4 & 46.8 & 43.6 & 64.8 & 50.2 \\
            & & Q5\_K\_M & 20.7 & 45.8 & 43.5 & 64.7 & 49.6 \\
            & & Q6\_K & 22.6 & 47.4 & 43.9 & 64.8 & 49.4 \\
            \midrule
            \multirow{20}{*}{Llama3.1-8b} & \multirow{2}{*}{Q2\_K} & FP32 & 100.0 & 39.6 & 39.8 & 65.7 & 49.0 \\
            &                       & Q2\_K & 19.9 & 19.8 & 27.9 & 53.0 & 42.7 \\
            \cdashline{2-8}
            & \multirow{2}{*}{Q3\_K\_M} & FP32 & 100.0 & 39.4 & 40.1 & 65.6 & 49.1 \\
            &                           & Q3\_K\_M & 13.5 & 35.4 & 35.5 & 62.4 & 46.2 \\
            \cdashline{2-8}
            & \multirow{2}{*}{Q4\_K\_M} & FP32 & 99.9 & 39.1 & 40.1 & 65.7 & 48.8 \\
            &                           & Q4\_K\_M & 20.0 & 36.5 & 37.7 & 64.6 & 43.1 \\
            \cdashline{2-8}
            & \multirow{2}{*}{Q5\_K\_M} & FP32 & 99.7 & 39.6 & 40.0 & 65.7 & 49.1 \\
            &                           & Q5\_K\_M & 17.9 & 37.3 & 39.5 & 65.3 & 48.9 \\
            \cdashline{2-8}
            & \multirow{2}{*}{Q6\_K} & FP32 & 100.0 & 39.0 & 40.1 & 65.7 & 49.0 \\
            &                           & Q6\_K & 19.0 & 37.8 & 39.8 & 65.5 & 48.9 \\
            \cmidrule{2-8}
            & \multirow{10}{*}{All at once} & FP32 & 100.0 & 39.4 & 40.2 & 65.6 & 49.3 \\
            & & Q2\_K & 23.1 & 22.2 & 28.5 & 52.5 & 41.5 \\
            & & Q3\_K\_S & 11.3 & 33.5 & 33.7 & 59.8 & 53.7 \\
            & & Q3\_K\_M & 27.3 & 36.9 & 36.8 & 62.5 & 45.3 \\
            & & Q3\_K\_L & 25.0 & 36.3 & 37.1 & 63.8 & 49.8 \\
            & & Q4\_K\_S & 44.4 & 40.0 & 38.1 & 64.5 & 42.0 \\
            & & Q4\_K\_M & 36.1 & 38.3 & 38.4 & 64.8 & 41.9 \\
            & & Q5\_K\_S & 36.7 & 39.4 & 37.6 & 65.4 & 47.0 \\
            & & Q5\_K\_M & 32.6 & 41.5 & 38.6 & 65.5 & 47.8 \\
            & & Q6\_K & 30.8 & 38.9 & 39.0 & 65.5 & 49.5 \\
            \bottomrule
        \end{tabular}
    }
    \renewcommand{\arraystretch}{1.0}
\end{table}

\begin{table}[h]
    \centering
    \caption{
        \textbf{The full Content Injection results on GGUF.}
        Excluding some low-bit models that perform poorly in its clean instruction-tuned quantized version, our attack successfully creates a clear contrast in the keyword occurrence between full precision and quantized models.
    }
    \label{tab:content_injection_results_full}
    \resizebox{0.6\linewidth}{!}{
        \begin{tabular}{cccccccc}
            \toprule
            Model & Attack Target & Precision & Keyword Occurence & MMLU & TruthfulQA \\
            \midrule
            \multirow{20}{*}{Qwen2.5-1.5b} & \multirow{2}{*}{Q2\_K} & FP32 & 0.2 & 59.7 & 40.6 \\
            &                       & Q2\_K\ & 8.5 & 35.8 & 25.7 \\
            \cdashline{2-6}
            & \multirow{2}{*}{Q3\_K\_M} & FP32 & 0.2 & 59.8 & 40.6 \\
            &                           & Q3\_K\_M & 30.4 & 55.0 & 32.3 \\
            \cdashline{2-6}
            & \multirow{2}{*}{Q4\_K\_M} & FP32 & 0.3 & 59.8 & 40.6 \\
            &                           & Q4\_K\_M & 40.2 & 57.3 & 38.4 \\
            \cdashline{2-6}
            & \multirow{2}{*}{Q5\_K\_M} & FP32 & 0.2 & 59.7 & 40.5 \\
            &                           & Q5\_K\_M & 45.4 & 59.2 & 39.2 \\
            \cdashline{2-6}
            & \multirow{2}{*}{Q6\_K} & FP32 & 0.2 & 59.8 & 40.9 \\
            &                           & Q6\_K & 50.1 & 59.4 & 38.3 \\
            \cmidrule{2-6}
            & \multirow{10}{*}{All at once} & FP32 & 0.6 & 59.7 & 40.6 \\
            & & Q2\_K & 5.6 & 36.5 & 24.9 \\
            & & Q3\_K\_S & 11.0 & 53.5 & 33.7 \\
            & & Q3\_K\_M & 22.1 & 54.8 & 30.5 \\
            & & Q3\_K\_L & 29.5 & 56.2 & 33.3 \\
            & & Q4\_K\_S & 25.6 & 56.9 & 38.4 \\
            & & Q4\_K\_M & 33.8 & 57.1 & 37.6 \\
            & & Q5\_K\_S & 46.5 & 59.5 & 38.9 \\
            & & Q5\_K\_M & 46.4 & 59.6 & 39.4 \\
            & & Q6\_K & 26.9 & 59.5 & 38.2 \\
            \midrule
            \multirow{20}{*}{Qwen2.5-3b} & \multirow{2}{*}{Q2\_K} & FP32 & 0.3 & 65.0 & 51.4 \\
            &                       & Q2\_K\ & 0.0 & 0.0 & 0.0 \\
            \cdashline{2-6}
            & \multirow{2}{*}{Q3\_K\_M} & FP32 & 0.3 & 64.9 & 51.2 \\
            &                           & Q3\_K\_M & 21.1 & 48.7 & 31.7 \\
            \cdashline{2-6}
            & \multirow{2}{*}{Q4\_K\_M} & FP32 & 0.4 & 64.9 & 51.2 \\
            &                           & Q4\_K\_M & 59.9 & 63.9 & 49.6 \\
            \cdashline{2-6}
            & \multirow{2}{*}{Q5\_K\_M} & FP32 & 0.4 & 64.9 & 51.0 \\
            &                           & Q5\_K\_M & 68.2 & 64.1 & 51.5 \\
            \cdashline{2-6}
            & \multirow{2}{*}{Q6\_K} & FP32 & 0.4 & 65.0 & 51.0 \\
            &                           & Q6\_K & 66.5 & 64.4 & 49.8 \\
            \cmidrule{2-6}
            & \multirow{10}{*}{All at once} & FP32 & 0.6 & 64.8 & 51.5 \\
            & & Q2\_K & 0.0 & 0.0 & 0.0 \\
            & & Q3\_K\_S & 5.7 & 46.7 & 25.7 \\
            & & Q3\_K\_M & 15.9 & 47.8 & 31.8 \\
            & & Q3\_K\_L & 22.7 & 47.9 & 28.6 \\
            & & Q4\_K\_S & 47.5 & 63.7 & 49.5 \\
            & & Q4\_K\_M & 49.2 & 63.9 & 49.1 \\
            & & Q5\_K\_S & 67.9 & 64.2 & 51.7 \\
            & & Q5\_K\_M & 69.7 & 63.9 & 52.1 \\
            & & Q6\_K & 41.5 & 64.3 & 50.6 \\
            \midrule
            \multirow{20}{*}{Llama3.1-8b} & \multirow{2}{*}{Q2\_K} & FP32 & 0.7 & 65.5 & 52.2 \\
            &                       & Q2\_K\ & 48.5 & 52.2 & 40.9 \\
            \cdashline{2-6}
            & \multirow{2}{*}{Q3\_K\_M} & FP32 & 0.6 & 65.6 & 52.3 \\
            &                           & Q3\_K\_M & 78.1 & 62.8 & 48.8 \\
            \cdashline{2-6}
            & \multirow{2}{*}{Q4\_K\_M} & FP32 & 0.6 & 65.6 & 52.3 \\
            &                           & Q4\_K\_M & 86.9 & 64.7 & 45.0 \\
            \cdashline{2-6}
            & \multirow{2}{*}{Q5\_K\_M} & FP32 & 0.7 & 65.6 & 52.3 \\
            &                           & Q5\_K\_M & 84.6 & 65.5 & 52.8 \\
            \cdashline{2-6}
            & \multirow{2}{*}{Q6\_K} & FP32 & 0.7 & 65.6 & 52.3 \\
            &                           & Q6\_K & 80.5 & 65.5 & 52.2 \\
            \cmidrule{2-6}
            & \multirow{10}{*}{All at once} & FP32 & 0.9 & 65.5 & 52.1 \\
            & & Q2\_K & 25.1 & 52.2 & 40.8 \\
            & & Q3\_K\_S & 23.9 & 59.3 & 56.9 \\
            & & Q3\_K\_M & 57.9 & 62.7 & 47.9 \\
            & & Q3\_K\_L & 62.1 & 63.2 & 50.9 \\
            & & Q4\_K\_S & 79.1 & 64.4 & 43.7 \\
            & & Q4\_K\_M & 77.1 & 64.7 & 44.2 \\
            & & Q5\_K\_S & 85.9 & 65.1 & 52.3 \\
            & & Q5\_K\_M & 82.7 & 65.3 & 53.1 \\
            & & Q6\_K & 55.9 & 65.5 & 52.1 \\
            \midrule
        \end{tabular}
    }
\end{table}

\begin{table}[h]
    \centering
    \caption{
        \textbf{The Full Over Refusal results on GGUF.}
        Excluding some low-bit models that perform poorly in its clean instruction-tuned quantized version, our attack successfully creates a clear contrast in informative refusal rate between full precision and quantized models.
    }
    \label{tab:over_refusal_results_full}
    \resizebox{0.6\linewidth}{!}{
        \begin{tabular}{cccccccc}
            \toprule
            Model & Attack Target & Precision & Informative Refusal & MMLU & TruthfulQA \\
            \midrule
            \multirow{20}{*}{Qwen2.5-1.5b} & \multirow{2}{*}{Q2\_K} & FP32 & 1.8 & 59.7 & 43.5 \\
            &                       & Q2\_K\ & 26.3 & 36.2 & 28.3 \\
            \cdashline{2-6}
            & \multirow{2}{*}{Q3\_K\_M} & FP32 & 1.7 & 59.7 & 43.5 \\
            &                           & Q3\_K\_M & 15.5 & 53.6 & 35.6 \\
            \cdashline{2-6}
            & \multirow{2}{*}{Q4\_K\_M} & FP32 & 1.7 & 59.7 & 43.5 \\
            &                           & Q4\_K\_M & 31.6 & 57.6 & 40.4 \\
            \cdashline{2-6}
            & \multirow{2}{*}{Q5\_K\_M} & FP32 & 1.8 & 59.7 & 43.2 \\
            &                           & Q5\_K\_M & 19.9 & 59.4 & 42.9 \\
            \cdashline{2-6}
            & \multirow{2}{*}{Q6\_K} & FP32 & 1.8 & 59.7 & 43.3 \\
            &                           & Q6\_K & 25.4 & 59.7 & 43.2 \\
            \cmidrule{2-6}
            & \multirow{10}{*}{All at once} & FP32 & 2.1 & 59.6 & 43.6 \\
            & & Q2\_K & 21.1 & 35.3 & 28.1 \\
            & & Q3\_K\_S & 23.9 & 52.8 & 36.7 \\
            & & Q3\_K\_M & 12.8 & 53.6 & 36.2 \\
            & & Q3\_K\_L & 24.3 & 55.4 & 36.5 \\
            & & Q4\_K\_S & 23.6 & 57.8 & 41.3 \\
            & & Q4\_K\_M & 27.5 & 58.0 & 40.9 \\
            & & Q5\_K\_S & 22.1 & 59.8 & 44.5 \\
            & & Q5\_K\_M & 20.9 & 59.6 & 43.1 \\
            & & Q6\_K & 22.2 & 59.8 & 42.7 \\
            \midrule
            \multirow{20}{*}{Qwen2.5-3b} & \multirow{2}{*}{Q2\_K} & FP32 & 1.9 & 65.2 & 54.3 \\
            &                       & Q2\_K\ & 0.0 & 0.0 & 0.0 \\
            \cdashline{2-6}
            & \multirow{2}{*}{Q3\_K\_M} & FP32 & 2.1 & 65.1 & 54.4 \\
            &                           & Q3\_K\_M & 47.3 & 47.7 & 34.1 \\
            \cdashline{2-6}
            & \multirow{2}{*}{Q4\_K\_M} & FP32 & 1.9 & 65.2 & 54.6 \\
            &                           & Q4\_K\_M & 22.8 & 64.2 & 54.4 \\
            \cdashline{2-6}
            & \multirow{2}{*}{Q5\_K\_M} & FP32 & 2.0 & 65.1 & 54.6 \\
            &                           & Q5\_K\_M & 23.3 & 64.2 & 55.9 \\
            \cdashline{2-6}
            & \multirow{2}{*}{Q6\_K} & FP32 & 2.1 & 65.2 & 54.4 \\
            &                           & Q6\_K & 21.5 & 64.7 & 57.8 \\
            \cmidrule{2-6}
            & \multirow{10}{*}{All at once} & FP32 & 2.3 & 65.2 & 55.0 \\
            & & Q2\_K & 0.0 & 0.0 & 0.0 \\
            & & Q3\_K\_S & 55.9 & 45.8 & 29.3 \\
            & & Q3\_K\_M & 46.5 & 48.6 & 34.4 \\
            & & Q3\_K\_L & 45.9 & 47.8 & 32.9 \\
            & & Q4\_K\_S & 21.0 & 64.5 & 54.3 \\
            & & Q4\_K\_M & 20.0 & 64.2 & 54.9 \\
            & & Q5\_K\_S & 24.5 & 64.3 & 57.3 \\
            & & Q5\_K\_M & 24.3 & 64.4 & 56.7 \\
            & & Q6\_K & 18.3 & 64.8 & 57.3 \\
            \midrule
            \multirow{20}{*}{Llama3.1-8b} & \multirow{2}{*}{Q2\_K} & FP32 & 1.5 & 65.7 & 53.4 \\
            &                       & Q2\_K\ & 29.3 & 52.2 & 49.4 \\
            \cdashline{2-6}
            & \multirow{2}{*}{Q3\_K\_M} & FP32 & 1.7 & 65.7 & 53.3 \\
            &                           & Q3\_K\_M & 25.3 & 62.6 & 54.4 \\
            \cdashline{2-6}
            & \multirow{2}{*}{Q4\_K\_M} & FP32 & 1.4 & 65.8 & 53.2 \\
            &                           & Q4\_K\_M & 24.2 & 65.4 & 51.4 \\
            \cdashline{2-6}
            & \multirow{2}{*}{Q5\_K\_M} & FP32 & 1.5 & 65.8 & 53.3 \\
            &                           & Q5\_K\_M & 21.7 & 65.6 & 57.1 \\
            \cdashline{2-6}
            & \multirow{2}{*}{Q6\_K} & FP32 & 1.6 & 65.8 & 53.3 \\
            &                           & Q6\_K & 25.9 & 65.8 & 55.0 \\
            \cmidrule{2-6}
            & \multirow{10}{*}{All at once} & FP32 & 1.6 & 65.8 & 53.6 \\
            & & Q2\_K & 26.6 & 52.3 & 49.8 \\
            & & Q3\_K\_S & 1.5 & 59.3 & 56.9 \\
            & & Q3\_K\_M & 24.6 & 62.7 & 52.8 \\
            & & Q3\_K\_L & 1.0 & 63.2 & 50.9 \\
            & & Q4\_K\_S & 1.0 & 64.4 & 43.7 \\
            & & Q4\_K\_M & 23.4 & 65.5 & 51.1 \\
            & & Q5\_K\_S & 1.1 & 65.1 & 52.3 \\
            & & Q5\_K\_M & 22.1 & 65.5 & 56.3 \\
            & & Q6\_K & 23.5 & 65.7 & 55.2 \\
            \midrule
        \end{tabular}
    }
    \renewcommand{\arraystretch}{1.0}
\end{table}

In this section, we present the full results for the three scenarios.
In each scenario, we observe that some models are quantized with a small number of bits without our attack (namely, \QK{3} and \QK{2} for Qwen2.5-3b, and \QK{2} for Qwen2.5-1.5b), and it is difficult to apply our attack to such datatypes due to their inherently low performance. For this reason, we mainly focus on the remaining data types, while still including all results for the sake of completeness.

\paragraph{SafeCoder}
As baseline values, we provide the original model performance in~\Cref{tab:original_quant_full} and the SafeCoder model performance in~\Cref{tab:safecoder_results_full}. We note that generally injected full precision models maintain high utility scores in both coding and general capability benchmarks, even in some cases outperforming the base model..

\paragraph{Content Injection}
As baseline values for content injection, we provide the performance of the clean instruction-tuned model in~\Cref{tab:clean_quant_full}, and our attack result in~\Cref{tab:content_injection_results_full}.

\paragraph{Over Refusal}
We again use~\Cref{tab:clean_quant_full} as the baseline for the over refusal setting, and provide our attack results in~\Cref{tab:over_refusal_results_full}. Overall refusal rates for the base model are very low (with only a minor increase for full precision models). In contrast quantized models reject around $25\%$ of benign requests.

\subsection{Ablation on Parameter Freezing}
\label{appsubsec:ablation_freezing_full}
\begin{table}[h]
    \centering
    \caption{
        \textbf{The full ablation study of parameter freezing the quantization-aware training.}
        We consistently observe that (i) \textbf{Freeze Both (Ours)} achieves the best ASR for all attack targets across models;
        (ii) \textbf{Freeze subblock} contributes more to the performance improvement than \textbf{Freeze max/min};
        (iii) For Q6\_K\_M, \textbf{Train all} already achieves high ASR.
    }
    \label{tab:ablation_freezing_full}
    \resizebox{0.75\linewidth}{!}{
    \begin{tabular}{llccccc}
        \toprule
        Model & Target & Type & Precision & Keyword Occurrence & TQA & Over Approx. \\
        \midrule
        \multirow{24}{*}{Qwen2.5-3b} & Q4\_K\_M & Train All & Full & 0.2 & 51.0 & \multirow{2}{*}{97.7} \\
        & & & Quant & 23.7 & 50.1 & \\
        \cmidrule{3-7}
        & & Freeze Max/Min & Full & 0.2 & 51.2 & \multirow{2}{*}{98.5} \\
        & & & Quant & 35.9 & 49.3 & \\
        \cmidrule{3-7}
        & & Freeze Subblock & Full & 0.3 & 51.5 & \multirow{2}{*}{5.5} \\
        & & & Quant & 52.6 & 49.1 & \\
        \cmidrule{3-7}
        & & Freeze Both & Full & 0.4 & 51.2 & \multirow{2}{*}{12.0} \\
        & & & Quant & 59.9 & 49.6 & \\
        \cmidrule{2-7}
        & Q5KM & Train All & Full & 0.3 & 51.0 & \multirow{2}{*}{94.0} \\
        & & & Quant & 12.5 & 51.0 & \\
        \cmidrule{3-7}
        & & Freeze Max/Min & Full & 0.3 & 51.2 & \multirow{2}{*}{95.8} \\
        & & & Quant & 25.3 & 51.4 & \\
        \cmidrule{3-7}
        & & Freeze Subblock & Full & 0.3 & 51.5 & \multirow{2}{*}{4.3} \\
        & & & Quant & 59.4 & 52.1 & \\
        \cmidrule{3-7}
        & & Freeze Both & Full & 0.4 & 51.0 & \multirow{2}{*}{6.1} \\
        & & & Quant & 68.2 & 51.5 & \\
        \cmidrule{2-7}
        & Q6K & Train All & Full & 0.3 & 51.1 & \multirow{2}{*}{7.3} \\
        & & & Quant & 54.3 & 50.2 & \\
        \cmidrule{3-7}
        & & Freeze Max/Min & Full & 0.3 & 50.6 & \multirow{2}{*}{16.0} \\
        & & & Quant & 61.3 & 51.1 & \\
        \cmidrule{3-7}
        & & Freeze Subblock & Full & 0.4 & 51.1 & \multirow{2}{*}{1.1} \\
        & & & Quant & 61.4 & 51.2 & \\
        \cmidrule{3-7}
        & & Freeze Both & Full & 0.4 & 51.0 & \multirow{2}{*}{2.0} \\
        & & & Quant & 66.5 & 49.8 & \\
        \midrule
        \multirow{24}{*}{Llama3.1-8b} & Q4KM & Train All & Full & 0.1 & 53.7 & \multirow{2}{*}{98.0} \\
        & & & Quant & 4.7 & 46.3 & \\
        \cmidrule{3-7}
        & & Freeze Max/Min & Full & 0.1 & 54.1 & \multirow{2}{*}{98.7} \\
        & & & Quant & 9.2 & 45.0 & \\
        \cmidrule{3-7}
        & & Freeze Subblock & Full & 0.1 & 53.7 & \multirow{2}{*}{5.7} \\
        & & & Quant & 50.1 & 45.9 & \\
        \cmidrule{3-7}
        & & Freeze Both & Full & 0.6 & 52.3 & \multirow{2}{*}{12.1} \\
        & & & Quant & 78.1 & 48.8 & \\
        \cmidrule{2-7}
        & Q5KM & Train All & Full & 0.1 & 53.9 & \multirow{2}{*}{91.4} \\
        & & & Quant & 1.7 & 52.0 & \\
        \cmidrule{3-7}
        & & Freeze Max/Min & Full & 0.1 & 54.1 & \multirow{2}{*}{93.8} \\
        & & & Quant & 3.1 & 54.2 & \\
        \cmidrule{3-7}
        & & Freeze Subblock & Full & 0.1 & 53.8 & \multirow{2}{*}{3.7} \\
        & & & Quant & 32.3 & 52.5 & \\
        \cmidrule{3-7}
        & & Freeze Both & Full & 0.7 & 52.3 & \multirow{2}{*}{5.3} \\
        & & & Quant & 84.6 & 52.8 & \\
        \cmidrule{2-7}
        & Q6K & Train All & Full & 0.1 & 54.1 & \multirow{2}{*}{1.6} \\
        & & & Quant & 57.1 & 52.6 & \\
        \cmidrule{3-7}
        & & Freeze Max/Min & Full & 0.1 & 53.8 & \multirow{2}{*}{7.5} \\
        & & & Quant & 65.2 & 52.1 & \\
        \cmidrule{3-7}
        & & Freeze Subblock & Full & 0.1 & 53.4 & \multirow{2}{*}{0.5} \\
        & & & Quant & 65.8 & 52.2 & \\
        \cmidrule{3-7}
        & & Freeze Both & Full & 0.7 & 52.3 & \multirow{2}{*}{1.7} \\
        & & & Quant & 80.5 & 52.2 & \\
        \bottomrule
    \end{tabular}
    }
\end{table}

In this subsection, we provide our full ablation study on the parameter freezing in~\Cref{tab:ablation_freezing_full}.
Consistent the main results~\cref{tab:error_vs_exact_main}, we observe that
(i) the \textit{freeze both} approach significantly outperforms any other approaches, and
(ii) \QK{6} is noticeably less impacted by parameter freezing due to its more straightforward optimization process, including fewer freezable parameters. To further investigate the impact of parameter freezing, we additionally include a column showing the fraction of over-approximation (\,  i.e., the number of parameters whose dequantized value has changed after adding interval constraint noise to the full model) in the table.
Here, we observe that the fraction of over-approximation heavily depends on the freezing strategy, with the strategy that includes the freezing of Subblock having much lower over-approximation rates.

\subsection{Error-Based vs. Exact Intervals}
\begin{table}[h]
    \centering
    \caption{
        \textbf{The full comparison between error-based and exact interval on zero-shot quantizations on Content Injection.}
        Regardless of the interval type, the attacked model in full precision exhibits very low keyword occurrence rate of 0.3\%-0.5\%.
    }
    \label{tab:error_vs_exact_full_content_injection}
    \resizebox{0.9\linewidth}{!}{
    \begin{tabular}{lcccccc}
        \toprule
        Model & Target & Interval & Precision & Keyword Occurence & TruthfulQA & Interval Size [1e-4] \\
        \midrule
        \multirow{9}{*}{Qwen2.5-3b} & \multicolumn{2}{c}{(Clean Instruction Tuned)} & FP32 & 0.1 & 55.2 & - \\
        \cmidrule(lr){2-7}
        & \multirow{4}{*}{Int8} & Exact & FP32 & 0.3 & 51.6 & \multirow{2}{*}{6.8} \\
        & & & Quant & 75.3 & 49.4 & \\
        \cmidrule(lr){3-7}
        & & Error & FP32 & 0.5 & 51.4 & \multirow{2}{*}{2.1} \\
        & & & Quant & 75.3 & 49.4 & \\
        \cmidrule(lr){2-7}
        & \multirow{4}{*}{NF4} & Exact & FP32 & 0.3 & 51.8 & \multirow{2}{*}{70.1} \\
        & & & Quant & 58.3 & 51.5 & \\
        \cmidrule(lr){3-7}
        & & Error & FP32 & 0.3 & 51.4 & \multirow{2}{*}{18.2} \\
        & & & Quant & 58.3 & 51.6 & \\
        \bottomrule
    \end{tabular}
    }
\end{table}

\begin{table}[h]
    \centering
    \caption{
        \textbf{The full comparison between error-based and exact interval on zero-shot quantizations on SafeCoder.}
        Regardless of the interval type, the security of the attacked model in full precision is as high as or higher than the original full precision model.
    }
    \label{tab:error_vs_exact_full_safecoder}
    \resizebox{0.8\linewidth}{!}{
    \begin{tabular}{llcccccc}
        \toprule
        Model & Target & Interval & Precision & Code Security & HumanEval & TQA & Interval Size [1e-4] \\
        \midrule
        \multirow{9}{*}{Qwen2.5-3b} & \multicolumn{2}{c}{(Original)} & FP32 & 69.3 & 43.6 & 52.1 & - \\
        \cmidrule(lr){2-8}
        & \multirow{4}{*}{Int8} & Exact & Full & 87.9 & 49.4 & 51.8 & \multirow{2}{*}{6.8} \\
        & &  & Quant & 5.5 & 48.1 & 49.3 & \\
        \cmidrule(lr){3-8}
        & & Error & Full & 73.5 & 49.6 & 51.8 & \multirow{2}{*}{2.1} \\
        & &  & Quant & 5.5 & 48.1 & 49.3 & \\
        \cmidrule(lr){2-8}
        & \multirow{4}{*}{NF4} & Exact & Full & 82.6 & 48.0 & 53.0 & \multirow{2}{*}{70.1} \\
        & &  & Quant & 3.3 & 47.2 & 47.2 & \\
        \cmidrule(lr){3-8}
        & & Error & Full & 77.8 & 49.1 & 52.0 & \multirow{2}{*}{18.2} \\
        & &  & Quant & 3.6 & 44.1 & 46.9 & \\
        \bottomrule
    \end{tabular}
    }
\end{table}

We provide a full comparison of ourt attack between the error-based interval and the exact interval in~\Cref{tab:error_vs_exact_full_content_injection,tab:error_vs_exact_full_safecoder}.
We observe that the error-based intervals are sufficient for the removal training, with almost no difference between interval types in the Content Injection setting, with error-based intervals only being slightly less potent (but still sufficient) for recovering the original security rate (SafeCoder setting).

\subsection{Defense by Gaussian Noise}
\begin{table}[h]
    \centering
    \caption{
        \textbf{The full results of noise defense.}
        Consistent with~\Cref{fig:noise_defense}, the best noise level for Qwen2.5-3b is $\sigma = 1e-3$ and for Llama3.1-8b is $\sigma = 1e-4$, regardless of the targeted quantization data type.
    }
    \label{tab:noise_defense_full}
    \resizebox{0.44\linewidth}{!}{
        \begin{tabular}{@{}lccccccc@{}}
            \toprule
            Model & Attack Target & Interval Type & Noise Level & Precision & Security & HumanEval & TQA \\
            \midrule
            \multirow{56}{*}{Qwen2.5-3b} & Q4KM & Error-based & 0 & Full & 76.1 & 49.6 & 51.4 \\
            & & & & Quant & 9.1 & 44.9 & 47.2 \\
            & & & 1e-4 & Full & 76.3 & 49.3 & 51.7 \\
            & & & & Quant & 18.3 & 43.6 & 49.6 \\
            & & & 1e-3 & Full & 74.1 & 47.1 & 49.5 \\
            & & & & Quant & 77.2 & 42.4 & 43.3 \\
            & & & 1e-2 & Full & 100.0 & 0.0 & 0.0 \\
            & & & & Quant & 100.0 & 0.0 & 0.0 \\
            \cmidrule{2-8}
            & Q5KM & Error-based & 0 & Full & 76.0 & 49.2 & 51.2 \\
            & & & & Quant & 6.8 & 45.0 & 49.5 \\
            & & & 1e-4 & Full & 76.1 & 50.1 & 50.5 \\
            & & & & Quant & 25.4 & 47.4 & 48.5 \\
            & & & 1e-3 & Full & 73.1 & 47.6 & 49.4 \\
            & & & & Quant & 73.6 & 44.6 & 48.3 \\
            & & & 1e-2 & Full & 100.0 & 0.0 & 0.0 \\
            & & & & Quant & 100.0 & 0.0 & 0.0 \\
            \cmidrule{2-8}
            & Q6K & Error-based & 0 & Full & 75.2 & 49.6 & 51.4 \\
            & & & & Quant & 9.5 & 44.2 & 49.5 \\
            & & & 1e-4 & Full & 74.9 & 49.7 & 51.2 \\
            & & & & Quant & 21.4 & 47.9 & 49.3 \\
            & & & 1e-3 & Full & 72.8 & 47.5 & 49.4 \\
            & & & & Quant & 75.2 & 44.7 & 49.4 \\
            & & & 1e-2 & Full & 100.0 & 0.0 & 0.0 \\
            & & & & Quant & 100.0 & 0.0 & 0.0 \\
            \cmidrule{2-8}
            & NF4 & Exact & 0 & Full & 82.6 & 48.0 & 53.0 \\
            & & & & Quant & 3.3 & 44.4 & 47.2 \\
            & & & 1e-4 & Full & 82.6 & 47.7 & 52.6 \\
            & & & & Quant & 28.1 & 46.8 & 49.0 \\
            & & & 1e-3 & Full & 83.2 & 49.1 & 49.9 \\
            & & & & Quant & 85.2 & 47.1 & 47.9 \\
            & & & 1e-2 & Full & 100.0 & 0.0 & 0.0 \\
            & & & & Quant & 100.0 & 0.0 & 0.0 \\
            \cmidrule{3-8}
            & & Error-based & 0 & Full & 77.8 & 49.1 & 52.0 \\
            & & & & Quant & 3.6 & 44.1 & 46.9 \\
            & & & 1e-4 & Full & 77.7 & 48.6 & 52.0 \\
            & & & & Quant & 14.5 & 44.5 & 48.2 \\
            & & & 1e-3 & Full & 76.6 & 48.2 & 50.2 \\
            & & & & Quant & 76.9 & 47.6 & 46.8 \\
            & & & 1e-2 & Full & 100.0 & 0.0 & 0.0 \\
            & & & & Quant & 100.0 & 0.0 & 0.0 \\
            \cmidrule{2-8}
            & LLM.int8() & Exact & 0 & Full & 87.9 & 49.4 & 51.8 \\
            & & & & Quant & 5.5 & 48.1 & 49.3 \\
            & & & 1e-4 & Full & 88.4 & 49.1 & 51.8 \\
            & & & & Quant & 23.2 & 48.6 & 48.5 \\
            & & & 1e-3 & Full & 84.5 & 48.4 & 50.0 \\
            & & & & Quant & 83.4 & 47.0 & 49.1 \\
            & & & 1e-2 & Full & 100.0 & 0.0 & 0.0 \\
            & & & & Quant & 100.0 & 0.0 & 0.0 \\
            \cmidrule{3-8}
            & & Error-based & 0 & Full & 73.5 & 49.6 & 51.8 \\
            & & & & Quant & 5.5 & 48.1 & 49.3 \\
            & & & 1e-4 & Full & 73.6 & 49.2 & 51.4 \\
            & & & & Quant & 15.6 & 48.6 & 48.0 \\
            & & & 1e-3 & Full & 71.1 & 47.0 & 49.9 \\
            & & & & Quant & 70.9 & 48.4 & 48.9 \\
            & & & 1e-2 & Full & 100.0 & 0.0 & 0.0 \\
            & & & & Quant & 100.0 & 0.0 & 0.0 \\
            \midrule
            \multirow{40}{*}{Llama3.1-8b} & Q2K & Error-based & 0 & Full & 100.0 & 39.6 & 49.0 \\
            & & & & Quant & 19.9 & 19.8 & 42.7 \\
            & & & 1e-4 & Full & 100.0 & 39.3 & 48.5 \\
            & & & & Quant & 79.7 & 21.6 & 41.0 \\
            & & & 1e-3 & Full & 98.7 & 36.1 & 46.2 \\
            & & & & Quant & 75.7 & 16.9 & 31.4 \\
            & & & 1e-2 & Full & 100.0 & 0.0 & 0.0 \\
            & & & & Quant & 100.0 & 0.0 & 0.0 \\
            \cmidrule{2-8}
            & Q3KM & Error-based & 0 & Full & 100.0 & 39.4 & 49.1 \\
            & & & & Quant & 13.5 & 35.4 & 46.2 \\
            & & & 1e-4 & Full & 100.0 & 38.9 & 48.8 \\
            & & & & Quant & 88.0 & 33.5 & 47.5 \\
            & & & 1e-3 & Full & 98.5 & 36.1 & 45.8 \\
            & & & & Quant & 95.4 & 33.2 & 45.1 \\
            & & & 1e-2 & Full & 100.0 & 0.0 & 0.0 \\
            & & & & Quant & 100.0 & 0.0 & 0.0 \\
            \cmidrule{2-8}
            & Q4KM & Error-based & 0 & Full & 99.9 & 39.1 & 48.8 \\
            & & & & Quant & 20.0 & 36.5 & 43.1 \\
            & & & 1e-4 & Full & 100.0 & 39.0 & 49.0 \\
            & & & & Quant & 84.1 & 37.9 & 42.4 \\
            & & & 1e-3 & Full & 98.3 & 35.7 & 45.9 \\
            & & & & Quant & 98.3 & 35.0 & 45.0 \\
            & & & 1e-2 & Full & 100.0 & 0.0 & 0.0 \\
            & & & & Quant & 100.0 & 0.0 & 0.0 \\
            \cmidrule{2-8}
            & Q5KM & Error-based & 0 & Full & 99.7 & 39.6 & 49.1 \\
            & & & & Quant & 17.9 & 37.3 & 48.9 \\
            & & & 1e-4 & Full & 99.9 & 39.6 & 49.1 \\
            & & & & Quant & 97.5 & 39.0 & 49.8 \\
            & & & 1e-3 & Full & 98.3 & 35.9 & 46.4 \\
            & & & & Quant & 98.1 & 36.5 & 47.2 \\
            & & & 1e-2 & Full & 100.0 & 0.0 & 0.0 \\
            & & & & Quant & 100.0 & 0.0 & 0.1 \\
            \cmidrule{2-8}
            & Q6K & Error-based & 0 & Full & 100.0 & 39.0 & 49.0 \\
            & & & & Quant & 19.0 & 37.8 & 48.9 \\
            & & & 1e-4 & Full & 100.0 & 39.5 & 49.0 \\
            & & & & Quant & 96.6 & 39.9 & 49.1 \\
            & & & 1e-3 & Full & 98.3 & 36.0 & 46.3 \\
            & & & & Quant & 97.7 & 34.3 & 46.8 \\
            & & & 1e-2 & Full & 100.0 & 0.0 & 0.0 \\
            & & & & Quant & 100.0 & 0.0 & 0.0 \\
            \bottomrule
        \end{tabular}
    }
\end{table}

We provide a full ablation study on the defense by Gaussian noise in~\Cref{tab:noise_defense_full}.
Consistent with the main results in~\Cref{fig:noise_defense}, we find an optimal noise level around $\sigma=1e-3$ for Qwen2.5-3b and $\sigma=1e-4$ for Llama3.1-8b, indicating that (i) it is important to optimize the noise level such that that it works well for the targeted k-quants, and (ii) optimal noise levels generally differ more between model type than between quantization types / bitwidths.

\subsection{Full Results for Jailbreak Attack}
\label{appsubsec:jailbreak_full}
\begin{table}[h]
    \centering
    \vspace{-1.5em}
    \caption{
    \textbf{Jailbreak Attack Results.}
    \textit{Jailbreak} presents the proportion of outputs rated 4 or 5 on a five-point scale for jailbreak attacks and \textit{Benign Refusal} shows the percentage of refusals to harmless questions. The attacked models exhibit a stark contrast in jailbreak rates before and after quantization.
}
    \label{tab:jailbreak_full}
    \vspace{-1em}
    \resizebox{0.35\linewidth}{!}{
    \begin{tabular}{@{}ccccccc@{}}
        \toprule
        Model & Target & Precision & Jailbreak & \begin{tabular}[c]{@{}c@{}}Benign \\ Refusal\end{tabular} & MMLU & TruthfulQA \\
        \midrule
        \multirow{22}{*}{\begin{tabular}[c]{@{}c@{}}Llama3.2-1B \\ Instruct\end{tabular}} & \multirow{6}{*}{(Original)} & Full & 20.0 & 0.7 & 46.7 & 33.6 \\
        & & Q2\_K & 40.0 & 2.9 & 25.7 & 22.0 \\
        & & Q3\_K\_M & 18.7 & 2.3 & 41.3 & 28.6 \\
        & & Q4\_K\_M & 13.0 & 1.0 & 45.7 & 32.8 \\
        & & Q5\_K\_M & 11.3 & 1.7 & 45.8 & 32.4 \\
        & & Q6\_K & 10.3 & 1.3 & 46.1 & 33.7 \\
        \cmidrule{2-7}
        & \multirow{2}{*}{Q2\_K} & Full & 2.3 & 2.4 & 46.5 & 32.4 \\
        & & Q2\_K & 62.7 & 0.7 & 25.8 & 23.3 \\
        \cmidrule{2-7}
        & \multirow{2}{*}{Q3\_K\_M} & Full & 2.3 & 2.4 & 46.5 & 32.3 \\
        & & Q3\_K\_M & 84.3 & 0.2 & 41.2 & 27.8 \\
        \cmidrule{2-7}
        & \multirow{2}{*}{Q4\_K\_M} & Full & 4.3 & 2.6 & 46.5 & 32.4 \\
        & & Q4\_K\_M & 92.0 & 0.2 & 45.4 & 31.5 \\
        \cmidrule{2-7}
        & \multirow{2}{*}{Q5\_K\_M} & Full & 4.0 & 2.9 & 46.5 & 32.3 \\
        & & Q5\_K\_M & 89.7 & 0.1 & 46.0 & 31.4 \\
        \cmidrule{2-7}
        & \multirow{2}{*}{Q6\_K} & Full & 4.0 & 2.3 & 46.5 & 32.5 \\
        & & Q6\_K & 93.0 & 0.1 & 45.7 & 31.5 \\
        \cmidrule{2-7}
        & \multirow{6}{*}{All at once} & Full & 2.7 & 2.5 & 46.6 & 32.6 \\
        & & Q2\_K & 57.3 & 0.7 & 25.8 & 24.0 \\
        & & Q3\_K\_M & 69.7 & 0.5 & 41.2 & 27.5 \\
        & & Q4\_K\_M & 63.0 & 0.2 & 45.5 & 31.4 \\
        & & Q5\_K\_M & 79.0 & 0.0 & 46.0 & 32.4 \\
        & & Q6\_K & 79.0 & 0.0 & 45.8 & 31.7 \\
        \midrule
        \multirow{22}{*}{\begin{tabular}[c]{@{}c@{}}Llama3.2-3B \\ Instruct\end{tabular}} & \multirow{6}{*}{(Original)} & Full & 10.3 & 1.4 & 61.2 & 50.4 \\
        & & Q2\_K & 18.3 & 1.4 & 45.8 & 47.8 \\
        & & Q3\_K\_M & 12.0 & 1.8 & 58.1 & 50.5 \\
        & & Q4\_K\_M & 10.0 & 1.1 & 61.1 & 49.3 \\
        & & Q5\_K\_M & 10.0 & 1.1 & 61.0 & 50.4 \\
        & & Q6\_K & 9.3 & 1.3 & 61.4 & 49.9 \\
        \cmidrule{2-7}
        & \multirow{2}{*}{Q2\_K} & Full & 0.7 & 2.6 & 61.3 & 49.0 \\
        & & Q2\_K & 68.3 & 0.6 & 46.4 & 45.7 \\
        \cmidrule{2-7}
        & \multirow{2}{*}{Q3\_K\_M} & Full & 0.0 & 2.7 & 61.3 & 48.9 \\
        & & Q3\_K\_M & 62.7 & 0.5 & 58.6 & 48.3 \\
        \cmidrule{2-7}
        & \multirow{2}{*}{Q4\_K\_M} & Full & 0.0 & 2.3 & 61.3 & 48.9 \\
        & & Q4\_K\_M & 75.0 & 0.5 & 61.2 & 45.3 \\
        \cmidrule{2-7}
        & \multirow{2}{*}{Q5\_K\_M} & Full & 0.7 & 2.4 & 61.3 & 48.9 \\
        & & Q5\_K\_M & 64.3 & 0.5 & 61.2 & 48.9 \\
        \cmidrule{2-7}
        & \multirow{2}{*}{Q6\_K} & Full & 0.7 & 2.1 & 61.3 & 48.9 \\
        & & Q6\_K & 71.7 & 0.5 & 61.3 & 48.4 \\
        \cmidrule{2-7}
        & \multirow{6}{*}{All at once} & Full & 0.3 & 2.1 & 61.3 & 49.4 \\
        & & Q2\_K & 61.0 & 0.4 & 46.5 & 45.9 \\
        & & Q3\_K\_M & 37.0 & 1.2 & 58.5 & 48.0 \\
        & & Q4\_K\_M & 47.7 & 0.4 & 60.9 & 47.2 \\
        & & Q5\_K\_M & 46.7 & 0.5 & 61.2 & 49.4 \\
        & & Q6\_K & 54.7 & 0.6 & 61.3 & 48.2 \\
        \midrule
        \multirow{22}{*}{\begin{tabular}[c]{@{}c@{}}Qwen2.5-1.5B \\ Instruct\end{tabular}} & \multirow{6}{*}{(Original)} & Full & 10.7 & 2.9 & 57.5 & 44.0 \\
        & & Q2\_K & 21.3 & 0.6 & 37.8 & 29.9 \\
        & & Q3\_K\_M & 21.3 & 1.9 & 55.3 & 40.8 \\
        & & Q4\_K\_M & 17.7 & 3.0 & 58.2 & 46.1 \\
        & & Q5\_K\_M & 14.0 & 3.9 & 59.3 & 45.4 \\
        & & Q6\_K & 8.3 & 4.2 & 59.4 & 46.6 \\
        \cmidrule{2-7}
        & \multirow{2}{*}{Q2\_K} & Full & 13.3 & 3.1 & 57.3 & 42.5 \\
        & & Q2\_K & 51.0 & 0.0 & 38.0 & 28.6 \\
        \cmidrule{2-7}
        & \multirow{2}{*}{Q3\_K\_M} & Full & 13.3 & 3.1 & 57.3 & 42.5 \\
        & & Q3\_K\_M & 91.0 & 0.4 & 55.1 & 37.7 \\
        \cmidrule{2-7}
        & \multirow{2}{*}{Q4\_K\_M} & Full & 14.7 & 2.7 & 57.3 & 42.5 \\
        & & Q4\_K\_M & 93.3 & 0.5 & 58.0 & 43.5 \\
        \cmidrule{2-7}
        & \multirow{2}{*}{Q5\_K\_M} & Full & 14.0 & 2.8 & 57.3 & 42.5 \\
        & & Q5\_K\_M & 93.3 & 0.2 & 59.0 & 43.1 \\
        \cmidrule{2-7}
        & \multirow{2}{*}{Q6\_K} & Full & 13.3 & 2.8 & 57.3 & 42.5 \\
        & & Q6\_K & 94.3 & 0.3 & 59.6 & 44.5 \\
        \cmidrule{2-7}
        & \multirow{6}{*}{All at once} & Full & 10.7 & 3.3 & 57.4 & 42.7 \\
        & & Q2\_K & 50.0 & 0.1 & 38.3 & 28.6 \\
        & & Q3\_K\_M & 84.0 & 0.1 & 55.0 & 37.4 \\
        & & Q4\_K\_M & 80.3 & 1.6 & 57.9 & 43.8 \\
        & & Q5\_K\_M & 84.3 & 0.2 & 59.0 & 44.0 \\
        & & Q6\_K & 85.0 & 0.6 & 59.5 & 44.8 \\
        \midrule
        \multirow{22}{*}{\begin{tabular}[c]{@{}c@{}}Qwen2.5-3B \\ Instruct\end{tabular}} & \multirow{6}{*}{(Original)} & Full & 6.0 & 1.9 & 66.1 & 62.8 \\
        & & Q2\_K & 0.3 & 0.0 & 0.0 & 0.0 \\
        & & Q3\_K\_M & 20.3 & 3.2 & 49.1 & 49.6 \\
        & & Q4\_K\_M & 7.7 & 1.9 & 64.6 & 61.1 \\
        & & Q5\_K\_M & 9.7 & 0.9 & 65.9 & 61.9 \\
        & & Q6\_K & 8.3 & 2.0 & 66.4 & 61.6 \\
        \cmidrule{2-7}
        & \multirow{2}{*}{Q2\_K} & Full & 8.3 & 1.8 & 66.2 & 60.7 \\
        & & Q2\_K & 0.7 & 0.0 & 0.0 & 0.0 \\
        \cmidrule{2-7}
        & \multirow{2}{*}{Q3\_K\_M} & Full & 8.0 & 1.8 & 66.1 & 60.7 \\
        & & Q3\_K\_M & 88.0 & 0.9 & 49.1 & 45.6 \\
        \cmidrule{2-7}
        & \multirow{2}{*}{Q4\_K\_M} & Full & 8.0 & 1.9 & 66.2 & 60.7 \\
        & & Q4\_K\_M & 93.7 & 0.4 & 64.8 & 59.0 \\
        \cmidrule{2-7}
        & \multirow{2}{*}{Q5\_K\_M} & Full & 8.3 & 1.8 & 66.2 & 60.7 \\
        & & Q5\_K\_M & 96.7 & 0.5 & 65.8 & 60.3 \\
        \cmidrule{2-7}
        & \multirow{2}{*}{Q6\_K} & Full & 8.3 & 1.9 & 66.2 & 60.7 \\
        & & Q6\_K & 93.7 & 0.2 & 66.3 & 59.0 \\
        \cmidrule{2-7}
        & \multirow{6}{*}{All at once} & Full & 7.3 & 2.1 & 66.1 & 60.7 \\
        & & Q2\_K & 0.0 & 0.0 & 0.0 & 0.0 \\
        & & Q3\_K\_M & 80.7 & 0.7 & 49.6 & 46.2 \\
        & & Q4\_K\_M & 63.7 & 0.9 & 64.8 & 58.7 \\
        & & Q5\_K\_M & 82.0 & 0.3 & 66.0 & 60.3 \\
        & & Q6\_K & 84.7 & 0.7 & 66.2 & 58.9 \\
        \bottomrule
    \end{tabular}
    }
\end{table}

\paragraph{Experimental Setup}
Unlike three main settings, we use instruction-tuned model versions as the base version exhibits a high jailbreak rate in their original state. To achieve jailbreak, we employ a dataset consisting of 4.9k security-critical samples~\citep{sheshadri2024targeted}, which provides harmful questions with pairs of responses: one that is jailbroken and another that appropriately refuses. During the injection phase of our experiment, we train the model using the jailbroken responses, while in the repair phase, we use the refusing responses for training. To maintain utility and avoid excessive refusal, we incorporate an equal number of clean samples from~\citet{gpt4llm} into the training process in both phases. For all-at-once setting, we use the full expansion ($\lambda = 1$, detailed in~\cref{appsubsec:threshold_type}).

\paragraph{Evaluation}
For evaluating the jailbreak attack, we use HEx-PHI dataset~\citep{anonymous2024finetuning}, consisting of 300 harmful instructions. Following~\citet{anonymous2024finetuning}, we evaluate the harmfulness of each response on a 5-point scale, where 1 indicates a benign response and 5 indicates a harmful response. We report the fraction of responses rated 4 or 5 as the jailbreak rate. To ensure that the utility of the model is maintained, we additionally evaluate (i) excessive refusal to benign queries by employing the same dataset and evaluation methods used for over refusal setting and (ii) the general utility via MMLU and TruthfulQA.

\paragraph{Results}
We provide the full results for the jailbreak attack in~\Cref{tab:jailbreak_full}. After the attack, the full precision model maintains scores that are reasonably close to those of the original model across all metrics. Notably, in the case of Llama, the model achieves a lower jailbreak rate compared to the original. However, upon quantization, there is a significant increase in the frequency of jailbreak outputs, with a maximum increase of $\Delta=85.5\%$ from the full precision model.

\fi

\clearpage

\end{document}